\documentclass[10pt,twocolumn,twoside]{IEEEtran}
\usepackage{times,cite,epsf,xspace,amsmath,amssymb,graphicx,subfigure,multirow,algorithm,algorithmic,tikz,verbatim,amsthm}
\usetikzlibrary{shapes,arrows}

\newcommand{\Amat}{{\bf A}}
\newcommand{\Bmat}{{\bf B}}
\newcommand{\Dmat}{{\bf D}}

\newcommand{\Imat}{{\bf I}}

\newcommand{\Smat}{{\bf S}}
\newcommand{\Tmat}{{\bf T}}
\newcommand{\Umat}{{\bf U}}
\newcommand{\Vmat}{{\bf V}}

\newcommand{\Xmat}{{\bf X}}
\newcommand{\Ymat}{{\bf Y}}
\newcommand{\Zmat}{{\bf Z}}

\newcommand{\evec}{{\bf e}}

\newcommand{\svec}{{\bf s}}

\newcommand{\xvec}{{\bf x}}
\newcommand{\yvec}{{\bf y}}
\newcommand{\zvec}{{\bf z}}
\newcommand{\zerovec}{{\bf 0}}

\newcommand{\vect}{\textrm{vec}}

\newcommand{\Atilde}{\tilde{{\bf A}}}

\newcommand{\Ytilde}{\tilde{{\bf Y}}}

\newcommand{\Gammamat}{{\bf \Gamma}}
\newcommand{\Sigmamat}{{\bf \Sigma}}
\newcommand{\epsilonvec}{{\bf \epsilon}}

\newtheorem{thm}{Theorem}
\newtheorem{propos}{Proposition}

\begin{document}
\title{Blind Compressed Sensing Over a\\Structured Union of Subspaces}

%\author{Jorge~Silva$^1$,~\IEEEmembership{Member,~IEEE,}
%	Minhua~Chen$^1$,
%	Yonina C. Eldar$^2$,~\IEEEmembership{Senior Member,~IEEE,}
%	Guillermo Sapiro$^3$,~\IEEEmembership{Senior Member,~IEEE,}
%	and Lawrence Carin$^1$,~\IEEEmembership{Fellow,~IEEE}\\
%$^1$Department of Electrical and Computer Engineering, Duke University\\
%$^2$Department Electrical Engineering, The Technion--Israel Institute of Technology\\
%$^3$Department of Electrical and Computer Engineering, University of Minnesota\\
%\{jg.silva,~minhua.chen\}@duke.edu,\\yonina@ee.technion.ac.il,~guille@umn.edu,~lcarin@ee.duke.edu
%}

\author{Jorge~Silva,~\IEEEmembership{Member,~IEEE,}
	Minhua~Chen,
	Yonina C. Eldar,~\IEEEmembership{Senior Member,~IEEE,}
	Guillermo Sapiro,~\IEEEmembership{Senior Member,~IEEE,}
	and Lawrence Carin,~\IEEEmembership{Fellow,~IEEE}%
\thanks{J. Silva is with the Department of Electrical and Computer Engineering,
Duke University, Durham, NC 27708-0291 USA (e-mail: jg.silva@duke.edu).}%
\thanks{M. Chen is with the Department of Electrical and Computer Engineering,
Duke University, Durham, NC 27708-0291 USA (e-mail: minhua.chen@ee.duke.edu).}%
\thanks{Y. C. Eldar is with the Department Electrical Engineering, The
Technion—Israel Institute of Technology, Haifa 32000, Israel (e-mail:
yonina@ee.technion.ac.il).}%
\thanks{G. Sapiro is with the Department of Electrical and Computer Engineering, University of Minnesota, Minneapolis, MN 55455 USA (e-mail: guille@umn.edu).}%
\thanks{L. Carin is with the Department of Electrical and Computer Engineering,
Duke University, Durham, NC 27708-0291 USA (e-mail: lcarin@ee.duke.edu).}}

\maketitle

\begin{abstract}
This paper addresses the problem of simultaneous signal recovery and dictionary learning based on compressive measurements. Multiple signals are analyzed jointly, with multiple sensing matrices, under the assumption that the unknown signals come from a union of a small number of disjoint subspaces. This problem is important, for instance, in image inpainting applications, in which the multiple signals are constituted by (incomplete) image  patches taken from the overall image. This work extends standard dictionary learning and block-sparse dictionary optimization, by considering compressive measurements (\emph{e.g.}, incomplete data). Previous work on blind compressed sensing is also generalized by using multiple sensing matrices and relaxing some of the restrictions on the learned dictionary. Drawing on results developed in the context of matrix completion, it is proven that both the dictionary and signals can be recovered with high probability from compressed measurements. The solution is unique up to block permutations and invertible linear transformations of the dictionary atoms. The recovery is contingent on the number of measurements per signal and the number of signals being sufficiently large; bounds are derived for these quantities. In addition, this paper presents a computationally practical algorithm that performs dictionary learning and signal recovery, and establishes conditions for its convergence to a local optimum. Experimental results for image inpainting demonstrate the capabilities of the method.
\end{abstract}

\section{Introduction}
The problem of learning a dictionary for a set of signals has received considerable attention in recent years. This problem is known to be ill-posed in general, unless constraints are imposed on the dictionary and signals. One such constraint is the assumption that the signals $\xvec_i\in\mathbb{R}^n, i=1,\dots,N$, can be sparsely represented under an unknown dictionary, $i.e.$, each vector $\xvec_i$ can be written as
\begin{align}
\xvec_i=\Dmat\svec_i+\epsilonvec_i,
\end{align}
where $\Dmat\in\mathbb{R}^{n\times r}$ is the dictionary, and $\svec_i\in\mathbb{R}^r$ are sparse coefficient vectors satisfying $\|\svec_i\|_0
\ll r$. The residual energy $\|\epsilonvec_i\|_2^2$ is assumed to be small. For example, \cite{aharon2006uniqueness} derives conditions to ensure the uniqueness of $\Dmat$ and the representations $\svec_i$, for the case in which all the coordinates of the signals $\xvec_i$ are observed. In this setting, the problem is known as dictionary learning (DL) or, sometimes, \emph{collaborative} DL to emphasize the fact that multiple signals are considered. Note that, in DL and in most related literature, ``uniqueness'' is defined up to an equivalence class involving permutations and rotations of the dictionary atoms; %(these are defined rigorously in the following section);
we follow the same convention throughout this paper.
In the DL setting, several algorithms have been proposed for estimating $\Dmat$ and $\svec_i$. Examples includes K--SVD \cite{aharon2006uniqueness} and the method of optimal directions (MOD) \cite{mairal2008discriminative}, both of which enjoy local convergence properties. More recently, \cite{wrightL1} has derived local convergence guarantees for $\ell_1$ minimization applied to DL when the signals are sparse.

A more structured type of sparsity is considered in block--sparse dictionary optimization (BSDO) \cite{rosenblum2010dictionary}, in which it is assumed that the nonzero coordinates of $\svec_i$ (active atoms) occur in blocks rather than in arbitrary locations. This property is called block sparsity \cite{eldar2010block}, and is important for the analysis of signals that belong to a union of a small number of subspaces, as described in \cite{eldar2009robust}.  The standard BSDO framework, like DL, assumes that  all coordinates of the signals $\xvec_i$ are observed. The block structure, $i.e.$, the number of atoms and composition of each block of the dictionary, is in general not known \emph{a priori} and should be estimated as part of the DL process. The BSDO algorithm reduces to K-SVD when the block size is one (standard sparsity).

While, technically, the set of all $k$--sparse signals in $\mathbb{R}^r$ is itself a union of $r \choose k$ subspaces, it greatly simplifies the problem if one considers the more structured setting of block sparsity. This reduces the number of subspaces to $L \choose K$, where $L$ is the total number of blocks in the dictionary and $K$ is the number of blocks that are active (block sparsity reduces to standard sparsity when all blocks are singletons). Block sparsity is closely related to the group LASSO in statistics literature \cite{meier2008group}, and also to the mixture of factor analyzers (MFA) statistical model, as noted in \cite{chen2009mfa}. The particular case for which only one block is active is called one--block sparsity and corresponds to a union of ${L \choose 1} = L$ possible subspaces, which is a dramatic simplification compared to the unstructured sparsity case. While this might at first sound limiting, the authors in \cite{yu2010solving} have obtained state--of--the--art image restoration results with a one--block sparsity model, and a variety of alternative methods have been proposed for this special case \cite{yu2010solving,yu2010statistical,elhamifar2009sparse}.

In applications we often do not have access to data $\xvec_i$, but only to compressive measurements
\begin{align}
	\yvec_i=\Amat\xvec_i=\Amat\Dmat\svec_i,
\end{align}
where $\Amat\in\mathbb{R}^{m\times n}$ has $m<n$ rows, so that $\yvec_i\in\mathbb{R}^m$. In order to enable recovery of $\svec_i$ from $\yvec_i$ even when $\Dmat$ is known, the sensing matrix $\Amat$ must satisfy incoherence properties with respect to $\Dmat$, as prescribed by compressed sensing (CS) theory  \cite{candes2006cs,donoho2006cs}. Learning $\Dmat$ from compressive measurements is called blind compressed sensing (blind CS), and does not admit a unique solution unless additional structure is imposed on $\Dmat$ \cite{gleichman2010blind}. See also the approach in \cite{duarte2009learning} for simultaneous dictionary learning and sensing matrix design in the CS scenario.

%\subsection{Application to image inpainting}
In this paper, we address simultaneous estimation of one--block--sparse signals and the corresponding dictionary, given only compressive measurements. This unifies blind CS and BSDO (for the one--block--sparse case in particular). It is known that the standard blind CS problem, where a fixed sensing matrix $\Amat$ is used, does not admit a unique solution in general, although a number of special cases of dictionaries for which such a solution exists have been identified in \cite{gleichman2010blind}. These special cases are: (i) finite sets of bases ($i.e.$, it is known that the dictionary is one member of a finite set of known bases for $\mathbb{R}^n$ ); (ii) sparse dictionaries ($i.e.$, the atoms of the dictionary themselves admit sparse representations) and (iii) block--diagonal dictionaries, with each block composed of orthogonal columns.

We are motivated in part by the problem of inpainting and interpolating an image \cite{bertalmio2000image,zhou2009non}, where one observes an incomplete image, \emph{i.e.}, we only know the intensity values at a subset of pixel locations (or a subset of their linear combinations). Additionally, the image is processed in (often overlapping) patches, which we convert to $n$--dimensional vectors (our signals of interest). We observe $m_i<n$ pixels in each patch, indexed by $i$, with these $m_i$ pixels selected at random. Therefore, the locations of the missing pixels are in general (at least partially) different for each patch. This means that, unlike the classical blind CS setting, the sensing matrix is not the same for all signals; we denote the sensing matrix for patch $i$ as $\Amat_i\in\mathbb{R}^{m_i\times n}$. The assumption of multiple and distinct $\Amat_i$ is crucial for solving the problem of interest, as it enables the use of existing results from matrix completion \cite{recht2009simpler}. Moreover, it has been demonstrated in \cite{yu2010statistical} that using multiple $\Amat_i$ improves inpainting and interpolation performance.

In conventional image inpainting, each $\Amat_i$ consists of a randomly chosen subset of $m_i$ rows of the $n\times n$ identity matrix (thereby selecting $m_i$ pixels). Successful estimation of the missing pixel intensity values is contingent on each patch being representable in terms of a small subset of the columns in the dictionary $\Dmat$ \cite{zhou2009non}. In other words, there needs to exist some $\Dmat$ such that the $\svec_i$ are (at least approximately) sparse. However, our analysis is not restricted to $\Amat_i$ being defined as in conventional inpainting. Other constructions typically used in CS, such as matrices with i.i.d. random--subgaussian \cite{rudelson2008sparse}  entries, or with rows drawn uniformly at random from an orthobasis \cite{candes2006cs}, can be employed. Moreover, in many settings it is appropriate to assume that the patches belong to a union of disjoint subspaces \cite{yu2010solving,ramirez2010classification}, with the number of subspaces being small. This motivates the one--block--sparsity assumption underlying our theoretical and computational results. This assumption, coupled with the use of multiple $\Amat_i$, allows us to ensure recovery of $\Dmat$ and the $\svec_i$ under milder conditions on $\Dmat$ and on the number of vectors $\yvec_i$, as compared to standard blind CS.

%\subsection{Contributions}
We show that unique recovery of the dictionary and the one--block--sparse signals is guaranteed with high probability, albeit with high computational effort, if the number of measurements per signal and the number of signals are sufficiently large. We derive algorithm--independent bounds for these quantities, thereby extending DL by considering compressive measurements and establishing a connection with matrix--completion theory.  Our results reduce to those known for DL when the signals are fully observed.

Additionally, we present a computationally feasible algorithm that performs DL and recovery of block--sparse signals based on compressive measurements with multiple sensing matrices. We automatically learn the block structure, in the same way as the BSDO algorithm \cite{rosenblum2010dictionary}; the size (number of atoms) and composition of each block is not known \emph{a priori} and we only need to specify a maximum block size. Our algorithm is closely related to BSDO, the main differences being that one--block sparsity and compressive measurements are considered. The estimates of $\Dmat$, $\svec_i$ and the block structure are found by alternating least-squares minimization.

It is shown that the algorithm converges to a local optimum under mild conditions, which we derive. This convergence analysis is not available for most other one--block--sparsity promoting methods, such as \cite{chen2009mfa}, or for most methods that rely on standard sparsity, such as \cite{zhou2009non}. An exception is the analysis in \cite{wrightL1} pertaining to $\ell_1$ minimization, where local convergence is proven when the number of measurements is at least $O(n^3 k)$ for signals with sparsity level $k$. However, compressive measurements are not considered. Besides our global uniqueness result, our method provably attains local convergence for, at most, $O(n k \log n)$ measurements, although our one--block sparsity assumption is stronger than in \cite{wrightL1}. The approach proposed in \cite{yu2010solving} also converges, but involves an even stronger assumption akin to intra--block sparsity, and does not include an algorithm--independent uniqueness analysis. Compelling experimental results are presented below, demonstrating the ability of our algorithm to perform inpainting of real images with a significant fraction of missing pixels.

%\subsection{Organization of the paper}
The remainder of the paper is organized as follows. Section \ref{sec:optim} presents preliminary definitions and a formulation of the optimization problem. Our uniqueness result is presented in Section \ref{sec:uniqueness}. Sections \ref{sec:multiple_A} and \ref{sec:conv_proof} respectively describe the proposed algorithm and the corresponding proof of convergence to a local optimum. Experimental results are described in Section \ref{sec:results} and concluding remarks are given in Section \ref{sec:conclusion}.

\section{Problem formulation}
\label{sec:optim}

\subsection{Preliminaries}
Assume vectors $\yvec_i\in\mathbb{R}^{m_i}, i=1,\dots,N$ are observed, such that
\begin{align}
    \yvec_i=\Amat_i\xvec_i=\Amat_i(\Dmat\svec_i+\epsilonvec_i)
    \label{eq:sensing}
\end{align}
where $\xvec_i\in\mathbb{R}^n$ is an unknown signal, $\Amat_i\in\mathbb{R}^{m_i\times n}$ is a known sensing matrix, $\Dmat\in\mathbb{R}^{n\times r}$ is an unknown dictionary, $\svec_i\in\mathbb{R}^r$ is an unknown sparse vector of coefficients such that $\xvec_i=\Dmat\svec_i+\epsilonvec_i$, with small residual $\|\epsilonvec_i\|_2$. We focus on the noiseless case, although our analysis can be extended straightforwardly to include observation noise, following \cite{candes2010noise}. Given $\yvec_i$ we would like to estimate $\xvec_i$ by finding $\Dmat$ and $\svec_i$. Thus, we are interested in solving 
\begin{align}
    \min_{\Dmat,\svec_1,\dots,\svec_N} &\sum_{i=1}^N \left\| \yvec_i-\Amat_i \Dmat\svec_i \right\|_2^2,\label{eq:ADS_opt}\\
 \textnormal{s.t.}~~~~&\svec_i~~\textnormal{one--block sparse}\nonumber
\end{align}
with our estimates denoted $\widehat{\Dmat}, \widehat{\svec}_1,\dots,\widehat{\svec}_N$.

It is assumed that each $\xvec_i$ lives in a subspace specified by a subset of the columns of $\Dmat$, so that there exists a permutation of the columns such that $\xvec_i=\Dmat\svec_i$ and the coefficient vector $\svec_i$ is one--block sparse, as defined below. It is also assumed that $\Dmat$ is composed of blocks (subsets of columns) corresponding to the blocks of $\svec_i$, with the cardinality of the blocks summing to $r$. The cardinality and composition of each block are unknown, although we fix a maximum cardinality as in BSDO. The atoms in each block are assumed orthonormal, in order to avoid scaling indeterminacies and also to obviate concerns with  sub--coherence (see \cite{eldar2010block} for a discussion). The expression in (\ref{eq:ADS_opt}), particularly the $\ell_2$ norm, is equivalent to a maximum-likelihood solution for the unknown model parameters, assuming that the components of $\epsilonvec_i$ are i.i.d. Gaussian and the model deviation is negligible.

\noindent {\bf Definition 1 (Block sparsity and one--block sparsity).} Let the dictionary $\Dmat\in\mathbb{R}^{n\times r}$ have a block structure such that $\Dmat=[\Dmat[1] \cdots \Dmat[L]]$, where each $\Dmat[\ell], \ell\in\{1,\dots,L\}$ is a unique subset of the columns in $\Dmat$, and the columns of $\Dmat[\ell]$ are assumed orthonormal for all $\ell$. Following \cite{eldar2010block}, a signal $\svec_i$ is $K$--block--sparse under dictionary $\Dmat$ if it admits a corresponding block structure $\svec_i=[\svec_i[1]^T~\cdots~\svec_i[L]^T]^T$ and if $\svec_i$ has zero entries everywhere except at a subset $\svec_i[\ell_1],\dots,\svec_i[\ell_K]$ of size $K\ll L$, corresponding to blocks $\Dmat[\ell_1],\dots,\Dmat[\ell_K]$ of $\Dmat$. Each dictionary block $\Dmat[\ell]\in\mathbb{R}^{n\times k_\ell}$ and each coefficient block $\svec_i[\ell]\in\mathbb{R}^{k_\ell}$. In general, the block size $k_\ell$ is not known \emph{a priori}, although it is common to define a maximum size $k_{\textnormal{max}}$.

We are specifically interested in one--block sparsity, where $K=1$. An interpretation of one--block sparsity is that it corresponds to signals that live in a union of $L$ linear subspaces, with each subspace spanned by the columns of one block $\Dmat[\ell]$. If block $\Dmat[\ell]$ has $k_\ell$ columns, then its subspace is of dimension $k_\ell$.  Figure \ref{fig:blocksparse} illustrates the concept of one--block sparsity. The index set of all signals which use block $\ell$ is $\omega_\ell=\{i: \svec_i[\ell]\neq \zerovec\}$. Note that for one--block--sparse signals there can be no overlap between $\omega_{\ell_1}$ and $\omega_{\ell_2}$ for $\ell_1\neq \ell_2$.%With fixed $\wvec$, the one--block--sparse case has similarities to the setting of \cite{chen2009mfa}, although for that work the dictionary learning and signal recovery steps are done separately and the signals live in a union of affine subspaces rather than a union of linear subspaces.
\begin{figure}[htb!]
    \centering
    \includegraphics[width=3.5in]{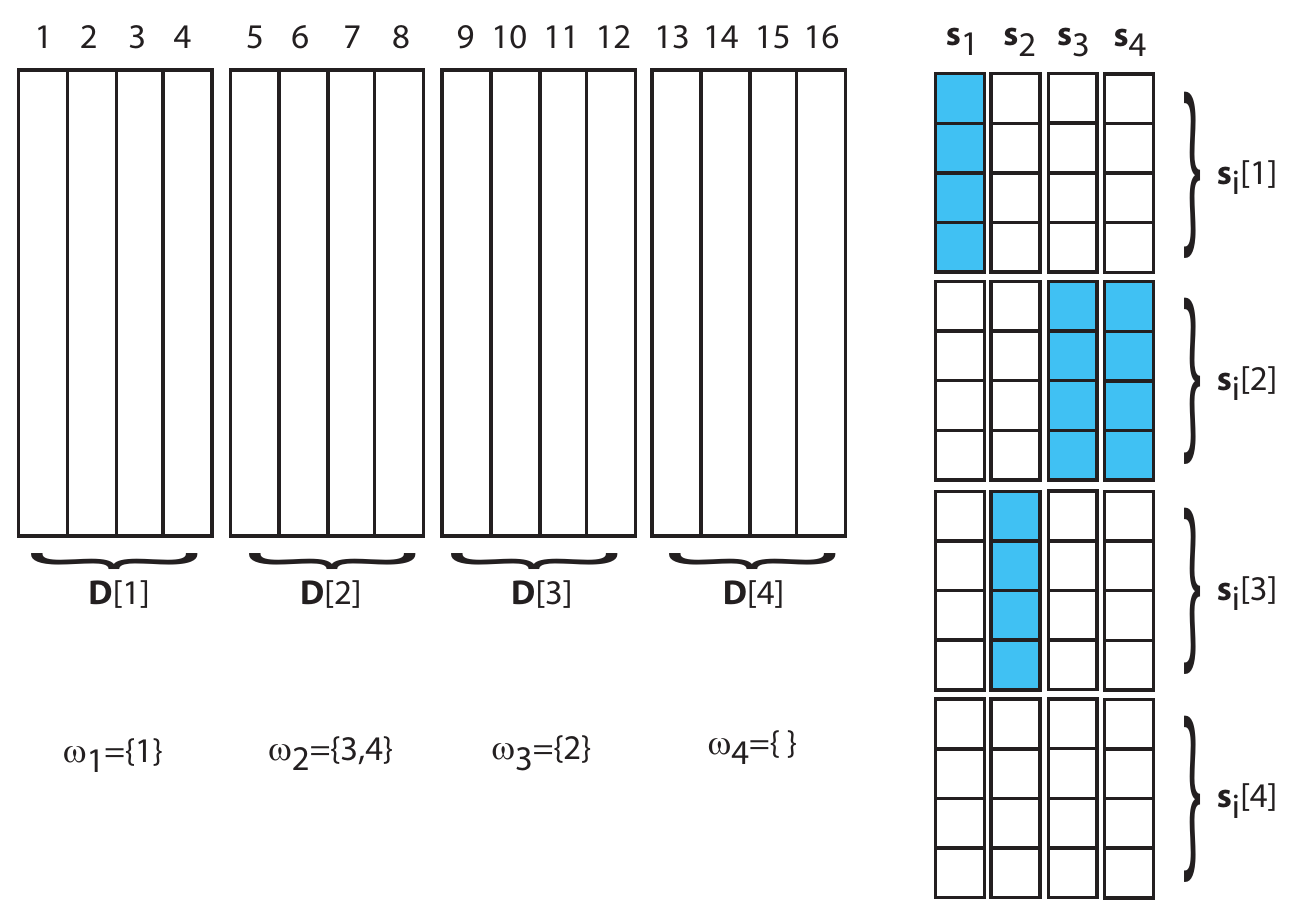}
    \caption{Illustration of the concept of one--block sparsity. The figure represents a dictionary with $r=16$ columns, arranged consecutively in blocks $\Dmat[1],\dots,\Dmat[4]$, all of size four. Also represented is a coefficient matrix with columns $\svec_1,\svec_2$, $\svec_3$ and $\svec_4$. Non--zero coefficient values are colored, while zero--valued coefficient values are blank. All signals are one--block sparse, as they only use one block each. Note that there exists no overlap between index sets $\omega_1, \omega_2$, $\omega_3$ and $\omega_4$.}
    \label{fig:blocksparse}
\end{figure}

The optimization problem (\ref{eq:ADS_opt}) can be decomposed as
\begin{align}
    \min_{\omega_\ell,\Dmat[\ell],\svec_1[\ell],\dots,\svec_N[\ell]} &~\sum_{i\in\omega_\ell}\left\| \yvec_i - \Amat_i \Dmat[\ell] \svec_i[\ell] \right\|_2^2,    \label{eq:ADS_opt_ell}\\
    \textnormal{s.t.} &~~~~~~~~\Dmat[\ell]^T\Dmat[\ell]=\Imat,\nonumber
\end{align}
for all $\ell$, where $\svec_i[\ell]\in\mathbb{R}^{k_\ell}$ is the $\ell$--th block of $\svec_i$. The solution to problem (\ref{eq:ADS_opt}) is only unique up to the following equivalence class.

\noindent {\bf Definition 2 (Equivalence class for $\widehat{\Dmat},\widehat{\svec}_1,\dots,\widehat{\svec}_N$).}
Given a solution $(\widehat{\Dmat},\widehat{\svec}_1,\dots,\widehat{\svec}_N)$ to (\ref{eq:ADS_opt}), with $\widehat{\Dmat}=[\widehat{\Dmat}[1] \cdots \widehat{\Dmat}[L]]$ and $\widehat{\svec}_i=[\widehat{\svec}_i[1]^T~\cdots~\widehat{\svec}_i[L]^T]^T$, $i=1,\dots,N$, any permutation $\ell_1^\prime,\dots,\ell_L^\prime$ of the indices $1,\dots,L$ corresponds to an equivalent solution $(\widehat{\Dmat}^\prime,\widehat{\svec}_1^\prime,\dots,\widehat{\svec}^\prime_N)$ of the form $\widehat{\Dmat}^\prime=[\widehat{\Dmat}[\ell_1^\prime] \cdots \widehat{\Dmat}[\ell_L^\prime]]$ and $\widehat{\svec}_i^\prime=[\widehat{\svec}_i[\ell_1^\prime]^T~\cdots~\widehat{\svec}_i[\ell_L^\prime]^T]^T$, $i=1,\dots,N$. Additionally, any sequence of invertible matrices $\Tmat_\ell\in\mathbb{R}^{k_\ell\times k_\ell}$, $\ell=1,\dots,L$, where $k_\ell$ is the size of the $\ell$--th block, also corresponds to an equivalent solution of the form $\widehat{\Dmat}^\prime=[\widehat{\Dmat}^\prime[1] \cdots \widehat{\Dmat}^\prime[L]]$ and $\widehat{\svec}_i^\prime=[(\widehat{\svec}_i[1]^\prime)^T~\cdots~(\widehat{\svec}_i[L]^\prime])^T]^T$, $i=1,\dots,N$, with $\Dmat[\ell]^\prime=\Dmat[\ell]\Tmat_\ell$ and $\svec_i[\ell]^\prime=\Tmat_\ell^{-1}\svec_i[\ell]$.

Our first goal is to develop conditions under which (\ref{eq:ADS_opt}) admits a unique solution. To this end we note that the index sets $\omega_\ell$ are unknown, and therefore estimates $\widehat{\omega}_\ell$ must be found. This is due to the fact that, \emph{a priori}, we do not know the correct assignments of columns to blocks. In our uniqueness analysis, we follow the same strategy as \cite{aharon2006uniqueness} and assume that all possible index sets can be tested. This is admittedly impractical in general, and it is done only for the purpose of constructing the uniqueness proof. The algorithm proposed in Section \ref{sec:multiple_A} provides estimates of $\omega_\ell$ without exhaustive search.

For a given $\omega_\ell$, (\ref{eq:ADS_opt_ell}) is an instance of the weighted orthonormal Procrustes problem (WOPP), which in general requires iterative optimization without guarantee of convergence to a global optimum \cite{mooijaart1990general}. In the following, we establish conditions under which uniqueness can be guaranteed. To this end we show that (\ref{eq:ADS_opt_ell}) can be written as a matrix--completion problem, and then rely on known results developed in that context. 

\subsection{Formulation as a matrix completion problem}
\label{sec:mc_formul}

To reformulate (\ref{eq:ADS_opt_ell}), define $\Xmat_{\omega_\ell}\in\mathbb{R}^{n\times |\omega_\ell|}$ as the matrix whose columns are the signals $\xvec_i$ that use block $\ell$, $i.e.$, with $i\in\omega_\ell$, and similarly define $\Smat_{\omega_\ell}[\ell]\in\mathbb{R}^{k_\ell\times |\omega_\ell|}$ as the matrix whose columns are the coefficient blocks $\svec_i[\ell]$ for $i\in\omega_\ell$. Since $\Dmat[\ell]$ has only $k_\ell$ columns and $\Smat_{\omega_\ell}[\ell]$ has $k_\ell$ rows, it is clear that $\Xmat_{\omega_\ell}=\Dmat[\ell]~\Smat_{\omega_\ell}[\ell]$ has rank at most $k_\ell$. In fact, unless the null space of $\Dmat[\ell]$ intersects the range of $\Smat_{\omega_\ell}$, the rank is exactly $k_\ell$. This should not happen in general, barring degeneracies that are precluded by the conditions assumed in our results. This suggests the strategy of treating each subproblem (\ref{eq:ADS_opt_ell}), assuming $\omega_\ell$ is known and hence the subspace selection has been determined (in our method we address estimation of $\omega_\ell$), as a low--rank matrix completion problem, as we illustrate next.

Define $\Atilde\in\mathbb{R}^{M\times n}$  with $M\geq n$, as the matrix constructed by taking the union of the unique rows of all sensing matrices $\Amat_i$, for $i\in\omega_\ell$ 
(in the interest of notation brevity, we omit the dependence of $\Atilde$ on $\ell$). For example, in the aforementioned inpainting problem, for which each $\Amat_i$ is defined by selecting rows at random from the $n\times n$ identity matrix, denoted $\Imat_{n\times n}$, we have $\Atilde$ equal to a subset of $\Imat_{n\times n}$ and in the limit, given enough $\Amat_i$, $\tilde{\Amat}=\Imat_{n\times n}$ (up to row permutation). However, the discussion below considers more-general $\Amat_i$, for example composed in terms of draws from a subgaussian distribution.  Let $\Ymat_{\omega_\ell}\in\mathbb{R}^{M\times |\omega_\ell|}$ be defined as $\Ymat_{\omega_\ell}=\Atilde~\Xmat_{\omega_\ell}$. When performing measurements, we do not observe all elements of $\Ymat_{\omega_\ell}$; for each column $i\in\omega_\ell$, we only have access to the entries selected by $\Amat_i$, these corresponding to the matching observed vector $\yvec_i$. This is illustrated in Figure \ref{fig:matrixobs}, which also shows a pictorial comparison with the cases of fully measured signals (DL) and a single measurement matrix $\Amat$ (standard blind CS).
\begin{figure*}[htb!]
    \centering
    \includegraphics[width=5.5in]{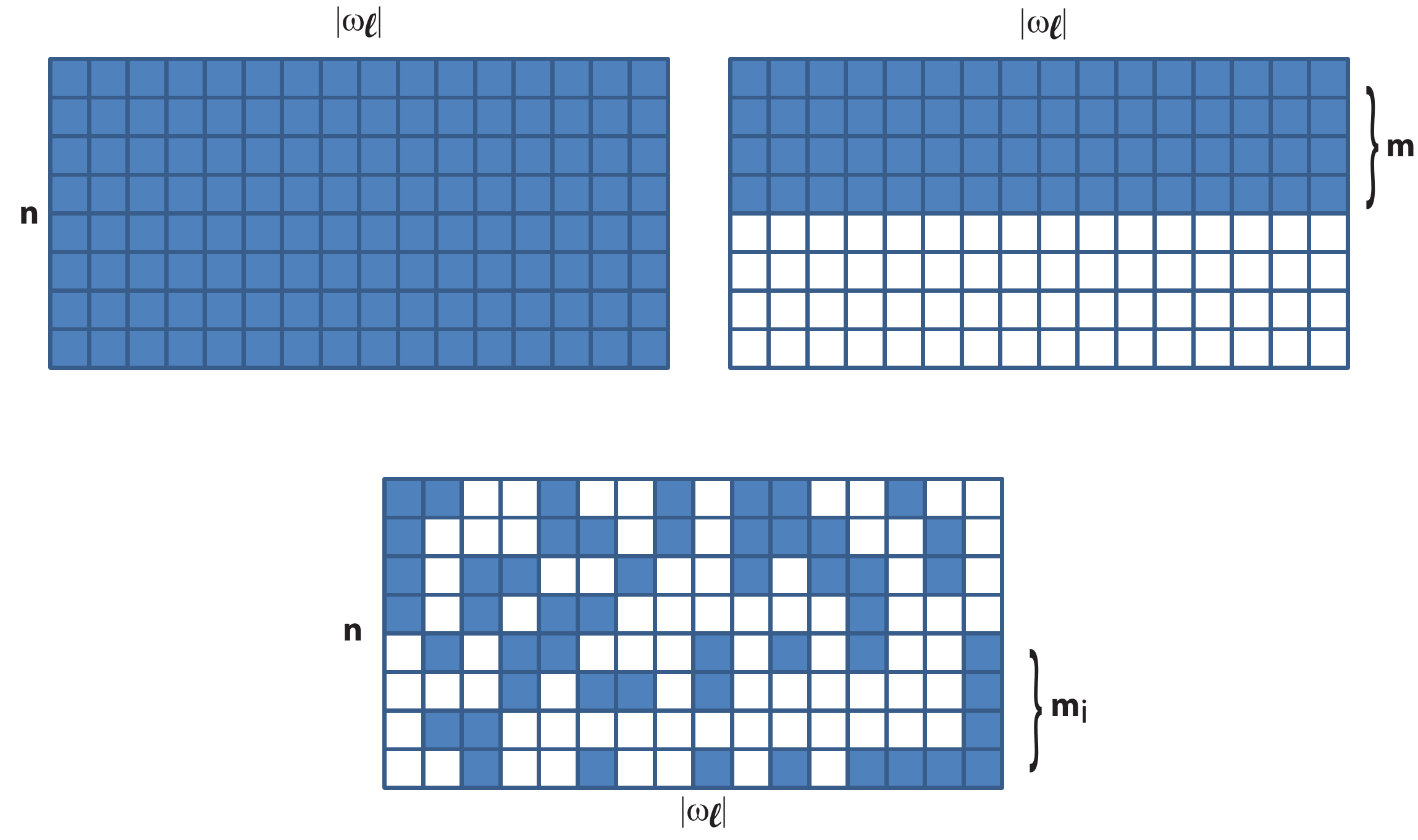}
    \caption{Illustration of the observed elements of $\Ymat_{\omega_\ell}$ with multiple measurement matrices (bottom), compared with the DL (top left) and standard blind CS (top right) cases. The colored squares represent observed elements, while blank squares represent unobserved elements of $\Ymat_{\omega_\ell}$. Here we assume that $\Ymat_{\omega_\ell}$ has $M=n$ rows.}
    \label{fig:matrixobs}
\end{figure*}

Denoting the locations of the observed entries of $\Ymat_{\omega_\ell}$ as
\begin{align}
	\Omega=\{(u,v): Y_{\omega_\ell}(u,v)~\textnormal{is observed}\},
\end{align}
the observation model can be written as
\begin{align}
    P_{\Omega}(\Ymat_{\omega_\ell})=P_{\Omega}(\Atilde\Xmat_{\omega_\ell})=P_{\Omega}(\Atilde\Dmat[\ell]\Smat_{\omega_\ell}[\ell]),
    \label{eq:sensing}
\end{align}
where $P_{\Omega}(\cdot)$ is the operator that extracts the values of its argument at locations indexed by $\Omega$. Each subproblem (\ref{eq:ADS_opt_ell}) can then be reformulated as
\begin{align}
    \min_{\omega_\ell,\Dmat[\ell],\Smat_{\omega_\ell}[\ell]} &~\left\| P_{\Omega}(\Ymat_{\omega_\ell})-P_{\Omega}(\Atilde\Dmat[\ell]\Smat_{\omega_\ell}[\ell]) \right\|_F    		\label{eq:MC_opt_ell}\\
    \textnormal{s.t.}&~~~~~~~~\Dmat[\ell]^T\Dmat[\ell]=\Imat,\nonumber
\end{align}

\noindent where $\|\cdot\|_F$ is the Frobenius norm.  If $\Xmat_{\omega_\ell}$ has rank $k_\ell$ (the block size), assuming $\tilde{\Amat}$ has rank greater or equal to $k_\ell$ and does not reduce the rank of $\Xmat_{\omega_\ell}$ (since $k_\ell$ is generally small and $\tilde{\Amat}$ is assumed to have rank $n$, these conditions for $\tilde{\Amat}$ are anticipated to be typical), we have that $\Ymat_{\omega_\ell}$ is also of rank $k_\ell$. This establishes the fact that, when $\Xmat_{\omega_\ell}$ has low rank, so does $\Ymat_{\omega_\ell}$ (matrix completion results are most useful when the rank is low, although they hold for any value of the rank). Each subproblem (\ref{eq:MC_opt_ell}) can then be 
solved by first completing the matrix $\Ymat_{\omega_\ell}$ and then obtaining $\Xmat_{\omega_\ell}$ and, subsequently, $\Dmat[\ell]$ and $\Smat_{\omega_\ell}[\ell]$. Note that, although this is a two--step process, if $\Atilde$ has rank $n$ then it can be inverted to find a unique mapping from any estimated complete matrix $\Ymat_{\omega_\ell}$ to a corresponding estimated $\Xmat_{\omega_\ell}$. As long as $\Xmat_{\omega_\ell}$ can be correctly estimated, the solution to (\ref{eq:MC_opt_ell}) can then be found by a singular value decomposition (SVD) of $\Xmat_{\omega_\ell}$. This is due to the fact that, for $\textrm{rank}(\Dmat[\ell]\Smat_{\omega_\ell})=k_\ell$, SVD minimizes the Frobenius norm of the residual $\Xmat_{\omega_\ell}-\Dmat[\ell]\Smat_{\omega_\ell}[\ell]$ w.r.t. $\Dmat[\ell]$ and $\Smat_{\omega_\ell}[\ell]$, by the Eckart-Young theorem \cite{eckart1936approximation}.

Following matrix completion theory \cite{recht2007guaranteed,candes2009exact,candes2010power,recht2009simpler}, if $\Ymat_{\omega_\ell}$ is truly of low rank (and additional technical conditions are met), it can be correctly completed with high probability by solving the convex program
\begin{align}
	\textnormal{minimize}~~~& \|\widehat{\Ymat}_{\omega_\ell}\|_* \label{eq:nuclear}\\
	\textnormal{s.t.}~~~&  P_{\Omega}(\widehat{\Ymat}_{\omega_\ell})=P_{\Omega}(\Ymat_{\omega_\ell}),\notag
\end{align}
where $\|\cdot\|_*$ is the \emph{nuclear norm}, which is defined for a generic matrix $\Zmat$ of rank $k$ as the sum of its singular values, $i.e.$,
\begin{align}
	\|\Zmat\|_*=\sum_{i=1}^k \gamma_i(\Zmat),
\end{align}
with $\gamma_i(\Zmat)$ indicating the $i$--th singular value of $\Zmat$. Importantly, problem (\ref{eq:nuclear}) is \emph{not} equivalent to (\ref{eq:MC_opt_ell}), since a solution of (\ref{eq:MC_opt_ell}) may not be a solution of  (\ref{eq:nuclear}). This is a reflection of the fact that there may exist a high--rank $\widehat{\Ymat}_{\omega_\ell}$ which solves  (\ref{eq:MC_opt_ell}) but has high nuclear norm. However, the opposite holds: a solution of  (\ref{eq:nuclear}) is always a solution of  (\ref{eq:MC_opt_ell}), and therefore of (\ref{eq:ADS_opt_ell}). Moreover, the solution of  (\ref{eq:nuclear}) will have low rank and therefore produce $\Dmat[\ell]$ with the smallest possible block size $k_\ell$, which is clearly desirable. Thus, the nuclear norm formulation yields those solutions of the original problem (\ref{eq:MC_opt_ell}) that are \emph{useful}.

%The construction (\ref{eq:nuclear}) finds the solution with minimum rank whose entries at locations $\Omega$ match those of $\Ymat_{\omega_\ell}$; therefore it also yields the block size $k_\ell$, which is the rank of the solution $\widehat{\Ymat}_{\omega_\ell}$.%For rank $k_\ell$,  $\widehat{\Ymat}_{\omega_\ell}$ is also the matrix closest to the true complete $\Ymat_{\omega_\ell}$, and therefore to $\Atilde\Xmat_{\omega_\ell}$, in the Frobenius norm.

In the following section, we state conditions on the number of observed entries and on $\Atilde$ for successful recovery of $\Ymat_{\omega_\ell}$, $\Xmat_{\omega_\ell}$ and, subsequently, of $\Dmat[\ell]$ and $\svec_1[\ell],\dots,\svec_N[\ell]$, assuming exhaustive search over $\omega_\ell$, by exploiting the connection between (\ref{eq:ADS_opt_ell}), (\ref{eq:MC_opt_ell}) and (\ref{eq:nuclear}).  In Section \ref{sec:multiple_A} we propose an algorithm that does not require exhaustive search and is guaranteed to converge to a local minimum under mild conditions here established. The notation used in this paper is summarized in Table 1 below, for ease of reference. %(and the conditions under which such an exhaustive search is guaranteed to yield the \emph{correct} partitioning of the data into appropriate dictionary blocks).
\begin{table}[htb]\caption{Quick notation guide}
	\begin{tabular}{rl}
		\hline
		$\yvec_i\in\mathbb{R}^{m_i}$	&	Observed vector $i$\\
		$m_i$		&	Number of observed coordinates of vector $\yvec_i$\\
		$\xvec_i\in \mathbb{R}^n$ 	&	Unknown signal $i$\\ 
		$\svec_i\in\mathbb{R}^r$ 	& 	Sparse coefficient vector $i$\\
		$\Dmat\in\mathbb{R}^{n\times r}$		&	Dictionary\\
		$\Dmat[\ell]\in\mathbb{R}^{n\times k_\ell}$	&	Dictionary block $\ell$, i.e. the $\ell$-th subset of the\\ & columns of $\Dmat$\\
		$k_\ell$		&	Number of atoms in block $\ell$\\
		$\omega_\ell=\{i: \svec_i[\ell]\neq \zerovec\}$	&	Set of indexes of the signals that use block $\ell$\\ & of the dictionary $\Dmat$ with coefficients $\svec_i$\\
		$|\omega_\ell|$ 		&	Number of signals that use block $\ell$\\
		$\Xmat_{\omega_\ell}\in\mathbb{R}^{n\times |\omega_\ell|}$	&	Subset of the signals associated with block $\ell$\\
		$\Smat_{\omega_\ell}\in\mathbb{R}^{r\times |\omega_\ell|}$		&	 Subset of the coefficient vectors\\ & associated with block $\ell$\\
		$\Smat_{\omega_\ell}[\ell]\in\mathbb{R}^{k_\ell\times |\omega_\ell|}$		&	 Block $\ell$ of $\Smat_\ell$, i.e. the $\ell$-th subset of the rows,\\ & corresponding to $\Dmat[\ell]$\\
		$\svec_i[\ell]$ 	&	Block $\ell$ of sparse coefficient vector $\svec_i$\\ & ($\ell$-th subset of the rows)\\
		$\Amat_i\in\mathbb{R}^{m_i\times n}$	&	Sensing matrix for vector $i$\\
		$\Atilde\in\mathbb{R}^{M\times n}$		&	Union of the rows of multiple sensing matrices\\
		$\Ymat_{\omega_\ell}\in\mathbb{R}^{M\times |\omega_\ell|}$		&	Incompletely observed data matrix associated\\ & with block $\ell$\\
		$\Omega$	& 	Observed locations of $\Ymat_{\omega_\ell}$\\
		$\otimes$	&	Kronecker product\\
		\hline
	\end{tabular}
\end{table}

\section{Uniqueness result}
\label{sec:uniqueness}
Our results rely on the concepts of \emph{spark} and \emph{coherence}, which we now define. The spark $\sigma(\cdot)$ of a matrix is the smallest number of columns that are linearly \emph{dependent}; the number of linearly \emph{independent} columns is the rank. Additionally, let $\mathcal{U}$ be a subspace of $\mathbb{R}^n$ with dimension $k$ and spanned by vectors $\{\zvec_u\}_{u=1,\dots,k}$. The coherence of $\mathcal{U}$ \emph{vis--\`{a}--vis} the standard basis $\{\evec_v\}_{v=1,\dots,n}$ is defined as \cite{recht2009simpler}
\begin{align}
\mu(\mathcal{U})=\frac{n}{k}\max_{u,v} \|\zvec_u^T \evec_v\|^2.
\end{align}
This quantity is very similar to the coherence of a matrix, as defined in \cite{tropp2004greed}. Our results also build on the following DL uniqueness conditions.

\noindent{\bf Definition 3 (Uniqueness conditions for DL \cite{aharon2006uniqueness}).}
\begin{itemize}
		\item {\bf Support:} $\|\svec_i\|_0=k<\frac{\sigma(\Dmat)}{2}, \forall i$
		\item {\bf Richness:} There exist at least $k+1$ signals for every possible combination of $k$ atoms from $\Dmat$. With regular sparsity level $k$ , this amounts to at least $(k+1){ r \choose k}$ signals. For one--block sparsity with $L$ blocks of size $k_\ell,\ell=1,\dots,L$, however, we need $\sum_{\ell=1}^L(k_\ell+1) $ signals, which is typically a far smaller number.
		\item {\bf Non--degeneracy:} Given $k+1$ signals from the same combination of $k$ atoms, their rank is exactly $k$ (general position within the subspace). Similarly, any $k+1$ signals from different combinations must have rank $k+1$.
\end{itemize}
These conditions apply to the case of fully measured signals and constitute a limiting case, since our problem reduces to DL when all $\Amat_i$ are equal to the identity matrix. Even in this limiting case, however, the one--block sparsity condition has the advantage of greatly reducing the number of possible subspaces and therefore the required number of signals, as stated in the richness property definition above. Note that we do not require prior knowledge of the block size $k_\ell$, as explained below. We now present our main result.

\begin{thm}[{\bf Uniqueness conditions for blind CS with multiple measurement matrices}]
Let $\yvec_i\in\mathbb{R}^{m_i}$, with $i=1,\dots,N$, be a set of observed vectors, obtained through projections of unknown signals $\xvec_i\in\mathbb{R}^n$ so that each $\yvec_i=\Amat_i\xvec_i$ and $\Amat_i\in\mathbb{R}^{m_i\times n}$ is a known sensing matrix. Furthermore, let $\xvec_i=\Dmat\svec_i$ where $\svec_i\in\mathbb{R}^r$ is an unknown vector of coefficients obeying one--block sparsity according to Definition 1 and $\Dmat\in\mathbb{R}^{n\times r}$ is an unknown dictionary having a block structure such that $\Dmat=[\Dmat[1] \cdots \Dmat[L]]$, with $\Dmat[\ell]^T\Dmat[\ell]=\Imat$.  Let $k_\ell$ be the number of columns of block $\Dmat[\ell]$ . In addition, let $\omega_\ell$ be the index set of the vectors $\xvec_i$ for which the corresponding $\svec_i$ have block $\svec_i[\ell]$ active, $i.e.$, nonzero, and let $|\omega_\ell|$ denote the number of such vectors. %Assume that the $\xvec_i, i\in\omega_\ell,$ are non--degenerate (any set of $k\leq k_\ell$ vectors span a subspace of rank at least $k$), and the corresponding $\svec_i[\ell], i\in\omega_\ell,$ have full support (all the coordinates are nonzero).
Define $\Atilde\in\mathbb{R}^{M\times n}$ as the union of the unique rows of all $\Amat_i$ for all $i\in\omega_\ell$, so that $\Ymat_{\omega_\ell}\in\mathbb{R}^{M\times |\omega_\ell|}$ with $M\geq n$ has columns given by $\Atilde\xvec_i$, for $i\in\omega_{\ell}$, and only a subset of the elements in $\Ymat_{\omega_\ell}$ are observed. Define $M_{1\ell}=\min(M,|\omega_\ell|)$, $M_{2\ell}=\max(M,|\omega_\ell|)$ and let $\Ymat_{\omega_\ell}=\Umat\Sigmamat\Vmat^T$ be the singular value decomposition of $\Ymat_{\omega_\ell}$. Define $\mu_\ell=\max(\mu_1^2,\mu_0)$, where $\mu_{1}$ is an upper bound on the absolute value of the entries of $\Umat\Vmat^T\sqrt{(M_{1\ell} M_{2\ell})/k_\ell}$ and $\mu_{0}$ is an upper bound on the coherence of the row and column spaces of $\Ymat_{\omega_\ell}$.

Then, by solving problem (\ref{eq:nuclear}) one can exactly recover all the blocks of $\Dmat$ and the coefficient vectors $\svec_i$ up to the equivalence class presented in Definition 2, with probability at least $1-6\log(M_2)(M_1+M_2)^{2-2\beta}-M_2^{2-2\sqrt{\beta}}$ for some $\beta>1$, if the following conditions are met for each $\ell\in\{1,\dots,L\}$.
\begin{itemize}
	\item[(i)] For all $i\in\omega_\ell$, $\|\svec_i\|_0=k_\ell<\frac{\sigma(\Atilde\Dmat)}{2}$.
	\item[(ii)] $|\omega_\ell|>k_\ell$.
	\item[(iii)] The vectors $\xvec_i, i\in\omega_\ell,$ are non--degenerate, \emph{i.e.}, any subset of $k\leq k_\ell$ vectors span a subspace of rank at least $k$.
	\item[(iv)] $32 \mu_\ell k_\ell (M_{1\ell}+M_{2\ell})\beta \log (2 M_{2\ell})$ entries of the matrix $\Ymat_{\omega_\ell}$ are observed uniformly at random. The total number of observed entries is, thus, $\sum_{\ell=1}^{L} 32 \mu_\ell k_\ell (M_{1\ell}+M_{2\ell})\beta \log (2 M_{2\ell})$.
\end{itemize}
\end{thm}

\noindent\begin{proof}[Proof of Theorem 1]
Condition (i)--(iii) are analogous to the DL uniqueness conditions in Definition 3, which are proven in \cite{aharon2006uniqueness}, and are always required even if all $\yvec_i$ are fully measured. Condition (i) is adapted to our setting by imposing the spark restriction on $\Atilde\Dmat$; this condition ensures that no two dictionary blocks are linearly dependent, and that the sensing matrices are sufficiently incoherent with respect to the dictionary. This is the case, with high probabilty, when $\Atilde$ obeys a random subgaussian or orthobasis construction \cite{candes2006cs,donoho2006cs}. 

Condition (ii) stems from the fact that any $k_\ell$--dimensional hyperplane is uniquely determined by $k_\ell+1$ vectors, and since a subspace will always contain the origin, only $k_\ell$ additional vectors are required. We write $|\omega_\ell|>k_\ell$ instead of $|\omega_\ell|\geq k_\ell$ because, in the setting of one--block sparsity, there exist $N=\sum_{\ell=1}^L |\omega_\ell|$ vectors from a union of $L$ different subspaces, and we need to be able to assign vectors to their respective subspaces.

If condition (iii) holds, which precludes spurious colinearities or coplanarities, then the assignment can be done by the (admittedly impractical) procedure of testing the rank of all possible $N \choose k_\ell+1$ matrices constructed by concatenating subsets of $k_\ell+1$ column vectors, as assumed in \cite{aharon2006uniqueness}. If the rank of such a matrix is $k_\ell$, then all its column vectors belong to the same subspace. Otherwise, the rank will be $k_\ell+1$. It is not necessary to know $k_\ell$ \emph{a priori}, provided that we test all possible values $k_\ell=1,\dots,n-1$ in the following way. Start by setting $k_\ell=1$ and testing all possible subsets containing $k_\ell+1=2$ signals. If any such subset has rank one, then we have found a (singleton) block. Any other possibility is precluded by the non--degeneracy condition (iii). If multiple subsets have rank one, then test the subspace angles in order to determine how many distinct singleton blocks exist. Record the corresponding blocks and signal assignments and remove the signals involved from further consideration (due to one--block sparsity, each signal belongs to only one block). Iterate this procedure until either: (a) we have exhausted the signals; (b) all $r$ atoms have been clustered into blocks, or (c) we have reached $k_\ell=n$.

%To see that Condition (i) is sufficient, note that if there exist any $k_\ell+1$ vectors from subspace $\ell$ that are fully observed, then a singular value decomposition (SVD) of only those $k_\ell+1$ vectors will reveal the fact that they come from the same $k_\ell$--dimensional subspace (because there will be only $k_\ell$ nonzero singular values) and will also yield an orthonormal basis for that subspace, as long as the vectors satisfy non--degeneracy. %In the inpainting problem, how likely are there to exist $k_\ell+1$ fully observed patches? Depending on the number and randomness of the missing pixels and the total number of patches, if $k_\ell$ is sufficiently low, then the probability can be quite high.
%This is case even if $\Atilde$ is not the identity matrix, as long as it is full-rank. In this case, one would perform the SVD of $\Atilde^+ \Ymat_{\omega_\ell}$ and obtain the same results as when $\Atilde=\Imat$. Here, $\Atilde^+$ is the pseudoinverse of $\Atilde$, and if $\Atilde$ has rank $n$, then we have $\Atilde^+\Atilde=\Imat$.

We now address Condition (iv). We have $|\omega_\ell|>k_\ell$ vectors from subspace $\ell$ and we  assume that $\Atilde$ has rank $n$ and is, therefore, invertible. As discussed in Section \ref{sec:optim}, $\Ymat_{\omega_\ell}$ is an incomplete low--rank matrix. Under conditions studied in \cite{recht2009simpler}, the convex program (\ref{eq:nuclear}) will recover the complete $\Ymat_{\omega_\ell}$ with high probability. Therefore, by showing uniqueness of the complete low--rank matrix $\Ymat_{\omega_\ell}$, we can show the uniqueness of the corresponding subspace spanned by $\Dmat[\ell]$, which can be found by first solving the optimization (\ref{eq:nuclear}) and then taking the estimate $\widehat{\Ymat}_{\omega_\ell}$ and performing the SVD of $\Atilde^+ \widehat{\Ymat}_{\omega_\ell}$.

We now invoke Theorem 1.1 in \cite{recht2009simpler} which states that, under the conditions and with the probability specified in our theorem, the minimizer of problem (\ref{eq:nuclear}) is unique and equal to the true $\Ymat_{\omega_\ell}$. This concludes the proof.
\end{proof}

Note that the aforementioned theorem from \cite{recht2009simpler}, like preceding results \cite{candes2010power}, is based on $\Omega$ having at least one entry for each row and for each column; this is already accounted for in the derivation of the bound on $|\Omega|$ and associated probability of successful recovery under a uniformly--at--random pattern for the missing data. If there is a fixed $\Amat_i=\Amat$ (with less than $n$ rows), or if there is any row or column entirely missing from $\Omega$, then there is no recovery guarantee. This is the case of standard blind CS, where there is a fixed $\Amat$ which is the same for all signals, and thus there are entire rows without observations, as illustrated in Figure \ref{fig:matrixobs}. Hence, blind CS requires additional constraints on $\Dmat$, as explained in \cite{gleichman2010blind}. Using multiple $\Amat_i$ allows us to avoid those constraints. Also, even when $M\ge n$, we are still subsampling because we do not observe all entries of $\Ymat_{\omega_\ell}$.

The limiting case when all blocks are singletons, $i.e.$, $k_\ell=1$ for all $\ell$, coupled with the one--block sparsity assumption, means that each patch belongs to a straight line in $\mathbb{R}^n$, which is a very strong restriction. On the other hand, if all $k_\ell=1$ but the one--block sparsity assumption is removed and there are $K$ active blocks (singleton atoms when $k_\ell=1$), then we revert to standard sparsity; Theorem 1 still applies, but we would need to be test $r \choose K$ possible combinations of atoms, which may be an extremely large number. In fact, the conditions in Theorem 1, with standard sparsity and additionally with all $\Amat=\Imat_{n\times n}$, reduce to those in DL. One--block sparsity makes Theorem 1 more appealing, since then there are only $L$ possible combinations instead of $r \choose K$.

Under Theorem 1, recovery is contingent on having the correct clustering given by $\omega_\ell$. In the above--presented analysis, it is assumed that we have the computational ability to exhaustively search over all index sets $\omega_\ell$. Under the conditions of Theorem 1, this search will achieve the correct assignment of signals to subspaces (clustering), and we need only concern ourselves with vectors that all come from the same subspace. This is done only for the purpose of constructing a proof, following the same strategy as in \cite{aharon2006uniqueness}. Also, like the DL uniqueness result in \cite{aharon2006uniqueness}, Theorem 1 is overly pessimistic in the required number of measurements. Specifically, for image inpainting applications, where $\Ymat_{\omega_\ell}$ typically has far more columns then rows, the prescribed number of measurements can be excessive (often orders of magnitude larger than the total number of pixels). This reflects a limitation of the current state--of--the--art in matrix completion theory, and we stress that other uniqueness results, such as \cite{aharon2006uniqueness} for the simpler DL setting, suffer from the same problem. Nevertheless, Theorem 1 provides peace of mind by guaranteeing the existence of a unique optimal solution given enough measurements. It remains necessary to address the practical issue of finding a solution with reasonable computational effort, from realistic amounts of unclustered data.

Toward that end, in the next section we propose an algorithm that can estimate the clustering and each of the subspaces without combinatorial search, although it is not guaranteed to reach the globally optimum solution. The algorithm uses alternating least squares, similarly to BSDO, K--SVD and MOD. It is computationally efficient and enjoyes local convergence guarantees given much fewer measurements than those prescribed by Theorem 1. In Section \ref{sec:conv_proof} we show convergence of the algorithm to a local optimum, as long as specified algorithm--specific conditions hold.%The algorithm is related to block--sparse dictionary optimization \cite{rosenblum2010dictionary} and K-SVD \cite{aharon2006uniqueness}, is capable of handling incomplete measurements and requires little tuning.

\section{Algorithm}
\label{sec:multiple_A}
Our algorithm expands on the iterative procedure described in \cite{rosenblum2010dictionary}. This procedure alternates between two major steps: ($i$) inferring the block structure and coefficients using sparse agglomerative clustering (SAC) \cite{rosenblum2010dictionary} and block--orthogonal matching pursuit (BOMP) \cite{eldar2010block}, and ($ii$) updating the dictionary blocks using block--K--SVD (BK--SVD) \cite{rosenblum2010dictionary}. The latter step assumes fully measured data. Briefly, SAC progressively merges pairs of blocks $\ell_1,\ell_2$ for which the index sets $\omega_{\ell_1}$ and $\omega_{\ell_2}$ are most similar, up to an upper limit ($k_{\max}$) on the block size. BOMP sequentially selects the blocks that best match the observed signals $\yvec_i$ and can be viewed as a generalization of OMP \cite{pati1993orthogonal,mallat1993matching} for the block--sparse case. Similarly, BK--SVD is an extension of K--SVD for block--sparsifying dictionaries. A more detailed explanation of these methods may be found in \cite{rosenblum2010dictionary} and the references therein.

We follow the procedure in \cite{rosenblum2010dictionary}, but replace the BK-SVD step by our Algorithm 1 in order to take compressive measurements into account, as the block diagram in Figure \ref{fig:block_diagram} illustrates. We employ alternating least-squares minimization to find $\Dmat$ and $\Smat$ in a manner similar to existing methods such as K--SVD or MOD but adapted to the CS setting.

\begin{figure}
\centering
\includegraphics[width=2.5in]{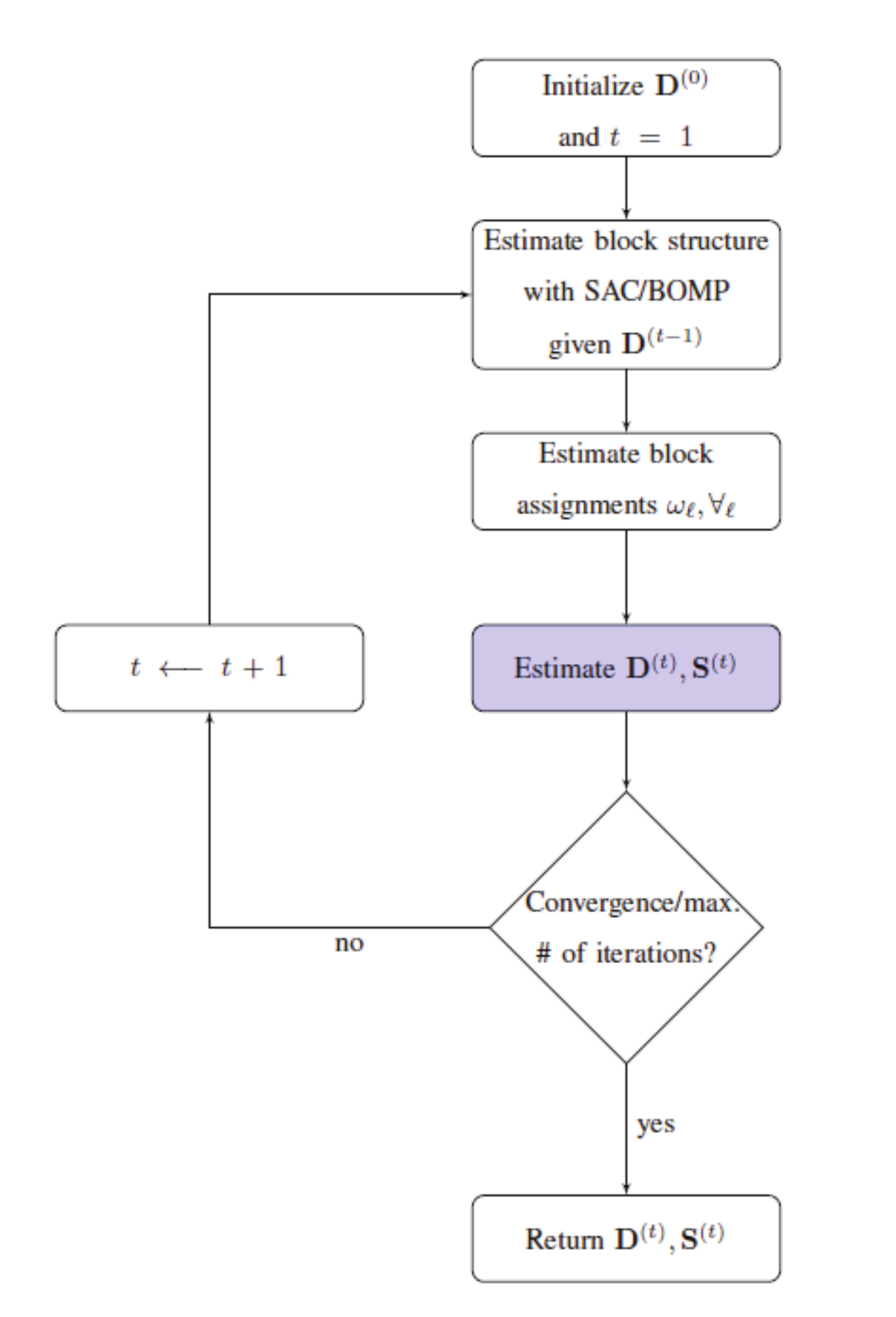}
\caption{Block diagram of the overall dictionary learning procedure. The analysis in this paper focuses on the shaded block, where we plug our Algorithm 1 instead of the BK-SVD step in \cite{rosenblum2010dictionary}.}
\label{fig:block_diagram}
\end{figure}

%\subsection{Algorithm 1: Estimation of $\Dmat[\ell]$ and $\Smat_{\omega_\ell}[\ell]$ via alternating minimization}
Recall that we are solving the optimization problem (\ref{eq:ADS_opt_ell}). Assuming the current estimate of $\omega_\ell$ and the block structure are fixed (as estimated in the previous SAC step), then we need to solve, for each $\ell$,
\begin{align}\label{Obj1}
     \min_{\Dmat[\ell],\Smat_{\omega_\ell}[\ell]} & \sum_{i\in \omega_\ell}\left\|\yvec_i-\Amat_{i}\Dmat[\ell]\svec_i[\ell]\right\|_2^{2}\\
	\textrm{s.t.} & ~~~~~\Dmat[\ell]^T\Dmat[\ell]=\Imat.\nonumber
\end{align}

\noindent The optimization alternates between $\Dmat[\ell]$ and $\Smat_{\omega_\ell}[\ell]$. For now we will set the constraint $\Dmat[\ell]^T\Dmat[\ell]=\Imat$ aside and focus on the unconstrained problem; we will return to the constraint later. A description of the two main algorithm steps is presented below, followed by the derivation of conditions for convergence.

\subsection{Dictionary update}
Holding $\svec_i$ fixed, we can analytically find $\Dmat[\ell]$ using properties of the Kronecker product and the vectorization operator, denoted as $\vect(\cdot)$. This is similar to MOD \cite{mairal2008discriminative} in that the entire dictionary block is updated at once. For a generic matrix $\Zmat$, $\vect(\Zmat)$ is the column vector constructed by vertically stacking the columns of $\Zmat$. For simplicity of presentation, and without loss of generality, assume that all $m_i=m$. Let $\Ytilde_{\omega_\ell}$ denote the matrix whose columns are the vectors $\yvec_i$, so that we have $|\vect(\Ytilde_{\omega_\ell})|=m|\omega_\ell|$. We rewrite the unconstrained problem over $\Dmat[\ell]$ as
\begin{align}\label{eq:vec}
\min_{\Dmat[\ell]}  \left\| \vect(\Ytilde_{\omega_\ell}) - \Bmat \vect(\Dmat[\ell]) \right\|_2^{2},
\end{align}
where $\Bmat$ is defined as
\begin{align}\label{eq:kronB}
	\Bmat=\left[\begin{array}{c} \svec_{i_1}[\ell]^T\otimes\Amat_{i_1} \\ \vdots\\ \svec_{i_{|\omega_\ell|}}[\ell]^T\otimes \Amat_{i_{|\omega_\ell|}}\end{array}\right],
\end{align}
with $\otimes$ denoting the Kronecker product. This is a least--squares optimization problem for a standard system of linear equations. The solution
\begin{align}
	\vect(\Dmat[\ell])=\Bmat^+ \vect(\Ytilde_{\omega_\ell}) \label{eq:D_update}
\end{align}
 is well--defined and unique (recall that we are holding the $\svec_i$, and therefore $\Bmat$, fixed) as long as $\Bmat$ has rank equal to $k_\ell n$, the number of unknowns in $\vect(\Dmat[\ell])$.

\subsection{Coefficient update}
After obtaining $\Dmat[\ell]$ for the current iteration, the coefficient vectors $\svec_{i}[\ell]$ are found analytically %by differentiating the objective (\ref{Obj1}) with respect to $\svec_i[\ell]$ and equaling to zero, yielding the least squares solution
by least squares as
\begin{align}
     %\Dmat[\ell]^T\Amat_i^T(\Amat_i\Dmat[\ell]\svec_{i}[\ell]-\xvec_i)&=\zerovec \label{eq:coeff_update} \\ \Rightarrow\nonumber
     \svec_{i}[\ell] &= (\Dmat^T[\ell]\Amat_i^T\Amat_{i}\Dmat[\ell])^{-1}(\Dmat^T[\ell]\Amat_i^T\xvec_i).\label{eq:coeff_update}
\end{align}
\noindent The factor $(\Dmat^T[\ell]\Amat_i^T\Amat_{i}\Dmat[\ell])^{-1}$ is invertible under conditions speficied in our convergence result below.
\subsection{Orthogonality constraint}
The orthogonality condition $\Dmat[\ell]^T\Dmat[\ell]=\Imat$ is currently imposed \emph{a posteriori} by performing a SVD decomposition of $\Dmat[\ell]$ and applying the resulting unitary rotation matrix to the $\svec_i[\ell]$. This is equivalent to Gram-Schmidt orthogonalization, and will not change the locations of the nonzero elements of the $\svec_i$, although the blocks $\Dmat[\ell]$ will not have orthogonal columns during the iterative procedure. This procedure has worked well in practice, and we therefore do not pursue a more direct constrained minimization. %We are currently researching an alternative method, based on optimization on a Stiefel manifold, which does maintain orthogonality during the entire optimization process.

In addition to the proposed Algorithm 1, we have also implemented and tested an alternative method where, instead of alternating least squares, we perform convex optimization of $\Dmat$ and $\Smat$ jointly (given the estimate of $\omega_\ell$ provided by the current SAC step) by completing matrix $\Ymat_{\omega_\ell}$, which is closer in spirit to the ideas underlying the uniqueness conditions in Theorem 1. However, existing general--purpose convex optimization packages do not readily scale up to matrices containing millions of entries as in our inpainting experiments, as noted in \cite{cai2008singular}. Using mixed nuclear/Frobenius norm minimization via the singular value thresholding (SVT) approach proposed in \cite{cai2008singular}, we have attained reconstruction performance comparable to our Algorithm 1, but at much higher computational cost. We have also found that successful convergence with this alternative SVT--based  approach depends on careful tuning of step--size and regularization parameters. In contrast, our alternating least squares method does not have any tuning parameters (other than the maximum block size $k_{\max}$) and is computationally far less demanding.

%The algorithm is summarized below.
\begin{algorithm}
    \caption{-- Joint estimation of $\Dmat[\ell]$ and $\Smat_{\omega_\ell}[\ell]$ via alternating minimization}
    \begin{algorithmic}
         \STATE
         \STATE{\bf Initialization:} Use the estimates of $\omega_\ell$, $\Smat_{\omega_\ell}$ and the block structure of $\Dmat$ from the preceding SAC/BOMP step
         \STATE
         \FORALL {$\ell$}
            \STATE {--} Form the matrix $\Bmat=\left[\begin{array}{c} \svec_{i_1}[\ell]^T\otimes\Amat_{i_1} \\ \vdots\\ \svec_{i_{|\omega_\ell|}}[\ell]^T\otimes \Amat_{i_{|\omega_\ell|}}\end{array}\right]$

            \STATE{--} Update block $\ell$ of the dictionary by computing $\vect(\Dmat[\ell])=\Bmat^+ \vect(\Ytilde_{\omega_\ell})$
	\STATE{--} Orthogonalize $\Dmat[\ell]$
            \STATE{--} Update the coefficients: $\svec_{i}[\ell] = (\Dmat^T[\ell]\Amat_i^T\Amat_{i}\Dmat[\ell])^{-1}(\Dmat^T[\ell]\Amat_i^T\xvec_i), \forall i\in\omega_\ell$
         \ENDFOR
    \end{algorithmic}
    \label{alg:multiple_A}
\end{algorithm}

\section{Convergence result}
\label{sec:conv_proof}

The following result establishes convergence conditions for our proposed alternating least--squares algorithm.

\begin{propos}[{\bf Convergence of Algorithm 1}]
Let $\yvec_i\in\mathbb{R}^{m_i}$, with $i=1,\dots,N$ be a set of vectors that satisfy the conditions stated in Theorem 1. Let $\Amat_i\in\mathbb{R}^{m_i\times N}$ be sensing matrices and $\xvec_i\in\mathbb{R}^n$ signals such that $\yvec_i=\Amat_i\xvec_i$.  Assume each signal is of the form $\xvec_i=\Dmat\svec_i$, with $\Dmat$ a dictionary comprised of blocks $\Dmat[1],\dots,\Dmat[L]$, where each block has $k_\ell$ atoms, and $\svec_i\in\mathbb{R}^r$ a one--block--sparse vector. Let $\omega_\ell$ be the index set of vectors $\xvec_i$ from dictionary block $\ell$, and let $\Ymat_{\omega_\ell}\in\mathbb{R}^{M\times |\omega_\ell|}$ be an incomplete data matrix such that $P_\Omega(\Ymat_{\omega_\ell})=P_\Omega(\Atilde\Dmat[\ell]\Smat_{\omega_\ell}[\ell])$, with $\Omega$ the locations of observed entries, and  $\Atilde\in\mathbb{R}^{M\times n}$ the union of the rows of the $\Amat_i$. Then, the alternating minimization procedure described in Section \ref{sec:multiple_A} will converge to a local minimum of (\ref{eq:ADS_opt}), with probability specified below, if the following conditions hold:
\begin{itemize}
	\item[(i)] For all $\ell$ and all $i\in\omega_\ell$, $\|\svec_i\|_0=k_\ell<\frac{\sigma(\Atilde\Dmat)}{2}$.
	\item[(ii)] For the subset of vectors associated with each block $\ell$, we have $|\omega_\ell|\geq n$.
	\item[(iii)] The total number of observed values, $|\Omega|$, is $O(k_\ell n)$ if the elements of all $\Amat_i$ are i.i.d Gaussian, and $O(k_\ell n \log n)$ if the rows of $\Amat_i$ are selected uniformly at random from $\Atilde$ and $\textnormal{rank}(\Atilde)=n$.
	\item[(iv)] Each vector $\yvec_i$ has $m_i\geq k_\ell$ measured values.
\end{itemize}
The probability of convergence, for each block $\ell$, is equal to one with i.i.d. Gaussian $\Amat_i$, and equal to $1-n^{-\beta+1}$ for random selection of rows from $\Atilde$, where $\beta$ is the constant such that $|\Omega|=\beta n\log n$.
\end{propos}

\begin{proof}[Proof of Proposition 1] Since condition (i) is the same as in Theorem 1, we focus on conditions (ii)--(iv). As mentioned above, the analytic solution (\ref{eq:D_update}) is well--defined and unique as long as $\Bmat$ has rank equal to $k_\ell n$.  Examining (\ref{eq:vec}) and (\ref{eq:kronB}) we can see that, if the $\svec_i[\ell]$ are non--degenerate ($i.e.$, rank($\Smat_{\omega_\ell}[\ell]$)=$k_\ell$, with $\Smat_{\omega_\ell}[\ell]$ the subset of nonzero rows of $\Smat_{\omega_\ell}$), then the rank condition is equivalent to having the number of observed elements $|\vect(\Ytilde_{\omega_\ell})|=|\Omega|\geq k_\ell n$ and also
\begin{align}
	\Gammamat=\left[\begin{array}{c}\Amat_{i_1}\\ \vdots\\ \Amat_{i_{|\omega_\ell|}} \end{array}\right],
\end{align}
that is, the vertical concatenation of all $\Amat_i$, having rank $n$. This is due to the fact that $\textrm{rank}(\svec_i \otimes \Amat_i)=\textrm{rank}(\svec_i)\times\textrm{rank}(\Amat_i),\forall i$. The number of observed values needed to ensure this rank condition depends on the probabilistic mechanism that generates the $\Amat_i$. We will analyze two such mechanisms: random Gaussian and random subset of projections. Without loss of generality, assume all $m_i=m$ (we can always treat $m$ as a lower limit on the $m_i$).

\noindent{\bf Random Gaussian:}
We assume that all $\Amat_i$  have elements independently drawn from a Gaussian distribution. Then, with probability one, any subset of size $n$ of the $m|\omega_\ell|$ rows of $\Gammamat$ are linearly independent. Therefore, the condition $m|\omega_\ell|\geq n$ ensures $\textrm{rank}(\Gammamat)=n$; since we need $m|\omega_\ell|\geq k_\ell n$ (and therefore $|\omega_\ell|>n$) in order to have enough observed entries anyway, the latter condition ensures that the rank of $\Gammamat$ is adequate. Note that this reasoning assumes that the $\Amat_i$ do not share rows; if they do share rows, \emph{e.g.}, in the case when there is a finite pool of random--Gaussian projections and the rows of $\Amat_i$ are subsets from the pool, then we must use the following result instead.

\noindent{\bf Random subset of projections:}
We analyze the case for which the rows of each sensing matrix $\Amat_i$ are randomly drawn from a pool of linearly independent rows, the union of which constitutes $\Atilde$. The typical example of this situation is when the rows of $\Atilde$ form a random orthobasis, although here orhonormality is not required (linear independence suffices). This includes the case when $\Atilde$ is the identity matrix, which is most relevant to the inpainting/interpolation problem. Note that the problem can be treated as an instance of the classic Coupon Collector problem (see, \emph{e.g.}, \cite{feller1968introduction}). There are $m|\omega_\ell|$ rows (the number of rounds our ``collector'' draws a randomly chosen row, or ``coupon'') that have to span a space of $n$ dimensions (the different types of ``coupons'' we need to collect). Denote as $Z_i^{m|\omega_\ell|}$ the event that no coupons of type $i$ have been drawn after $m|\omega_\ell|$ rounds. It is easy to see that
\begin{equation}
    P[Z_i^{m|\omega_\ell|}]\leq\left(1-\frac{1}{n}\right)^{m|\omega_\ell|},
\end{equation}
where the inequality comes from our ignoring, for simplicity, the fact that the rows are drawn without replacement within each sensing matrix (stated otherwise, there are no repeated rows within each $\Amat_i$). Then, by the union bound,
\begin{equation}
	P\left[\bigcup_{i=1,\dots,n}Z_i^{m|\omega_\ell|}\right] \leq n \left(1-\frac{1}{n}\right)^{m|\omega_\ell|}.
	\label{eq:bound}
\end{equation}
The left--hand side is precisely the probability that $\Gammamat$ is not rank $n$. To look at how this bound scales relative to $n$, we can apply the inequality $1-x\leq e^{-x}$ to obtain $\left(1-\frac{1}{n}\right)^{m|\omega_\ell|}\leq e^{-\frac{m|\omega_\ell|}{n}}$. Then, consider $m|\omega_\ell|=\beta n\log n$, with $\beta$ constant, and define $T$ to be the minimum value of $m|\omega_\ell|$ before we achieve full rank, so that $P[T>m|\omega_\ell|]=P[T>\beta n\log n]=P\left[\bigcup_{i=1,\dots,n}Z_i^{m|\omega_\ell|}\right]$. We have
\begin{align}
	P[T>\beta n\log n] \leq n \left(1-\frac{1}{n}\right)^{m|\omega_\ell|} \leq n  e^{-\frac{m|\omega_\ell|}{n}} = n e^{-\beta n \log n/n}=n^{-\beta+1}.
\end{align}

\noindent Therefore, we need $O(n\log n)$ observed values in order to achieve rank $n$ with arbitrarily high probability. Hence, we need $O(n k_\ell \log n)$ for solving (\ref{eq:vec}) for each block. Although we do not impose the orthogonality contraint directly in the minimization, any solution of  (\ref{eq:vec}) can be orthonormalized and thus become a solution of (\ref{Obj1}) \emph{without changing the subspace} defined by $\Dmat[\ell]$.

Concerning the coefficient update expression (\ref{eq:coeff_update}), the issue of invertibility comes up due to the term $(\Dmat^T[\ell]\Amat_i^T\Amat_{i}\Dmat[\ell])^{-1}$. Note that the matrix being inverted has size $k_\ell\times k_\ell$. The rank of $\Amat_i^T\Amat_{i}$ is equal to $m_i$, the number of rows of $\Amat_i$. For invertibility of the above expression, we need to have $k_\ell \leq m_i$, so that there is no loss of rank. This means that our method requires a lower bound on $m_i$.

We conclude the proof by invoking, as in \cite{aharon2006uniqueness}, the fact that all of the steps of this alternating minimization method are optimal under the above assumptions, and hence cannot increase the value of the objective function. Therefore, since the objective function is bounded from below by zero, the algorithm will converge.
\end{proof}

We have proven convergence in the objective function, which is a weak form of convergence in the sense that the objective function will attain a minimum but the estimates might not reach a stopping point. However, convergence in the objective function is still a useful result, and it is the same guarantee as K--SVD and MOD. In practice, the above result ensures that the required number of signals scales linearly with $k_\ell$ and also linearly (up to a log factor in the case of measurements using an orthobasis ensemble) with signal dimensionality $n$. These are more favorable bounds than those in Theorem 1 for uniqueness. However, this is not surprising since Theorem 1 is based on matrix completion theory, which establishes slightly pessimistic sufficient conditions for strong unique global recovery guarantees, while Proposition 1 only establishes convergence of our algorithm to a local minimum. It is also interesting that, in the case of random subsets of projections, the $O(k_\ell n\log n)$ requirement matches the information--theoretic limit established in \cite{candes2010power} for matrix completion. This similarity is due to the fact that, as in \cite{candes2010power} and other work, we make use of the Coupon Collector model. Regarding the lower bound on $m_i$, since we typically expect the block size to be small, the condition $m_i\geq k_\ell$ is relatively mild and should only be of concern for extremely small measurement fractions. %For concreteness, in our inpainting experiments, we have been using $k_\ell=4$ and patch size $n= 8\times 8 = 64$. We therefore expect the method to fail for measurement fractions under $k_\ell/n = 4/64 \approx 0.06$.

\section{Experimental results}
\label{sec:results}
The algorithm proposed in Section \ref{sec:multiple_A} is validated by inpainting the well--known ``Barbara'' ($512\times 512$ pixels) and ``house'' ($256 \times 256$ pixels) images, for varying percentages of observed pixels (these two examples are representative of many others we have considered). The images are processed in $8\times 8$ overlapping patches, treated as vectors of dimension $n=64$. In Figures \ref{fig:barb} and \ref{fig:house} we present the original images and the test versions, with 25\%, 50\% and 75\% of the pixel values observed (selected uniformly at random). In all experiments, the total number of dictionary elements is set to $r=256$. Figures \ref{fig:barbara512_8} and \ref{fig:house_4} show the inpainting results achieved by our algorithm from 50\% observed pixels, using maximum block sizes $k_{\max}=4$ and $k_{\max}=8$. The peak signal--to--noise ratios (PSNR) for this case are shown in Table \ref{tab:psnr}.
\begin{table}[h]
	\caption{Peak signal--to--noise ratios (PSNR) in inpainting tasks with 50\% observed pixels}
	\centering
	\begin{tabular}{lccc}
		& $k_{\max}=4$ & $k_{\max}=8$\\ \hline
		``Barbara"  & 27.68 dB & 27.93 dB\\
		``House"  & 31.80 dB & 32.03 dB
	\end{tabular}
	\label{tab:psnr}
\end{table}

\noindent The PSNR values for varying percentages of observed pixels are plotted in Figure \ref{fig:psnr}, for $k_{\max}$ equal to four and eight. Results are averaged over ten runs with different random locations for the missing pixels, with error bars showing the one standard deviation interval. These experiments are intended to validate our method, rather than claiming outperformance of the state--of--the--art in image inpainting and interpolation. Our performance is comparable to that of the algorithm described in \cite{zhou2009non}, which we have used on the same images. However, our model and algorithm are significantly simpler and, unlike \cite{zhou2009non}, our algorithm has convergence guarantees. Also, while the PSNR achieved in our experiments is lower than in the state--of--the--art approach in \cite{yu2010solving}, the results are not directly comparable due to the fact that additional structure is assumed in \cite{yu2010solving}, including fast decay of the singular values within each block (which is analogous to internal sparsity).%; moreover, \cite{yu2010solving} performs careful initialization of the dictionary, while we simply use random initialization.

Like other authors (\emph{e.g.}, \cite{wrightL1}), we have observed a phase transition phenomenon where the reconstruction performance falls sharply for measurement percentages under a rank--dependent threshold (near 40\% observed pixels in our case). This is shown in Figure \ref{fig:phase}, which was obtained using as input subsets of pixels from a union--of--subspaces approximation of the ``Barbara'' image for block sizes four and eight, thus matching our model exactly. We plot the empirical frequency of successful reconstructions over ten realizations of the missing pixel locations, for varying measurement percentages. The reconstruction is deemed successful when the PSNR exceeds 40 dB.

In addition, the learned dictionaries for $k_{\max}=4$ and $k_{\max}=8$ are depicted in Figures \ref{fig:dict_barbara} and \ref{fig:dict_house}. It is possible to observe that the atoms within each dictionary block are more similar to each other than those in different blocks. This is expected and desired clustering behavior, as one would like the blocks to reflect different image properties. This behavior is further illustrated in Figure \ref{fig:usage}, where the index of the block used by each patch is represented by a corresponding color (to avoid clutter, this is shown for the most frequently used dictionary blocks). It is apparent that patches with similar texture tend to share blocks.

In all cases tested, convergence occurs after a maximum of ten iterations. A non--optimized MATLAB implementation of Algorithm 1, including the SAC/BOMP steps, requires approximately 12 hours for inpainting the full $512\times 512$ ``Barbara'' image (worst case) on a computer with CPU clock frequency of 2.52 GHz. This computation time is similar to the BSDO algorithm described in \cite{rosenblum2010dictionary}. We employ overlapping patches in order to avoid block artifacts, leading to a total of $N=255,025$ vectors being processed for the ``Barbara'' image, with $N=62,001$ for the ``house'' image. Note that the overlapping procedure follows the same ``non--overlapping sensing with overlapping reconstruction'' strategy as \cite{yu2010statistical,zhou2009non}; this procedure employs a distinct sensing matrix per \emph{non--overlapping} patch, with the image reconstruction performed by averaging the \emph{overlapping} estimated patches.

\begin{figure*}[htb!]
    \centering
    \begin{tabular}{c}
    \includegraphics[width=4in]{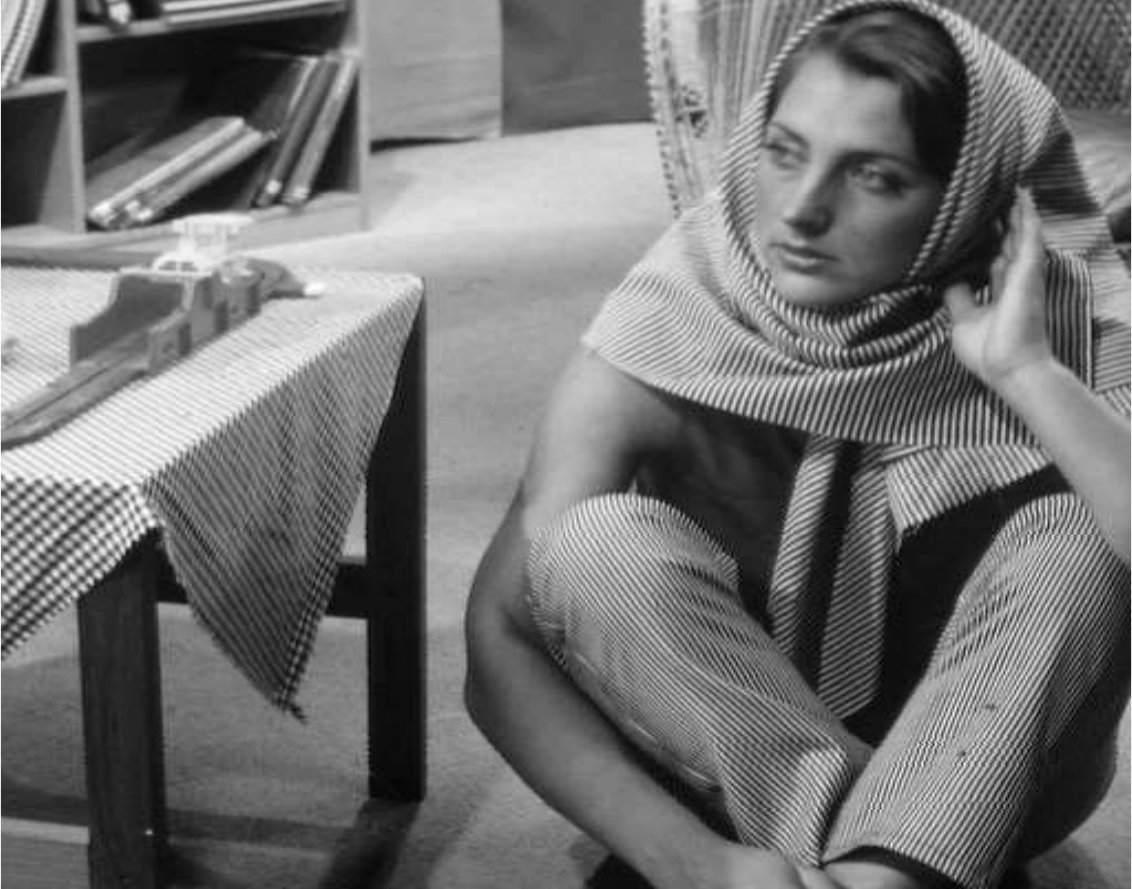} \\
    \includegraphics[width=2in]{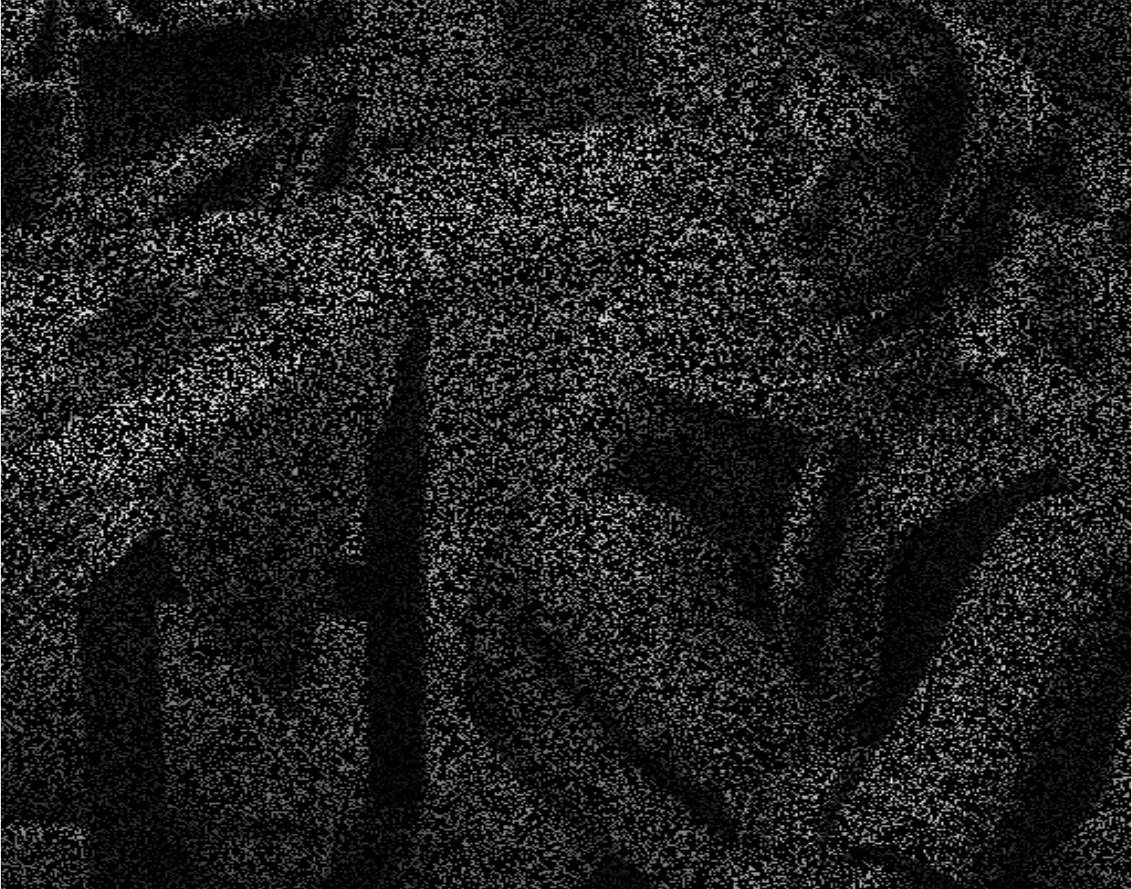} 
    \includegraphics[width=2in]{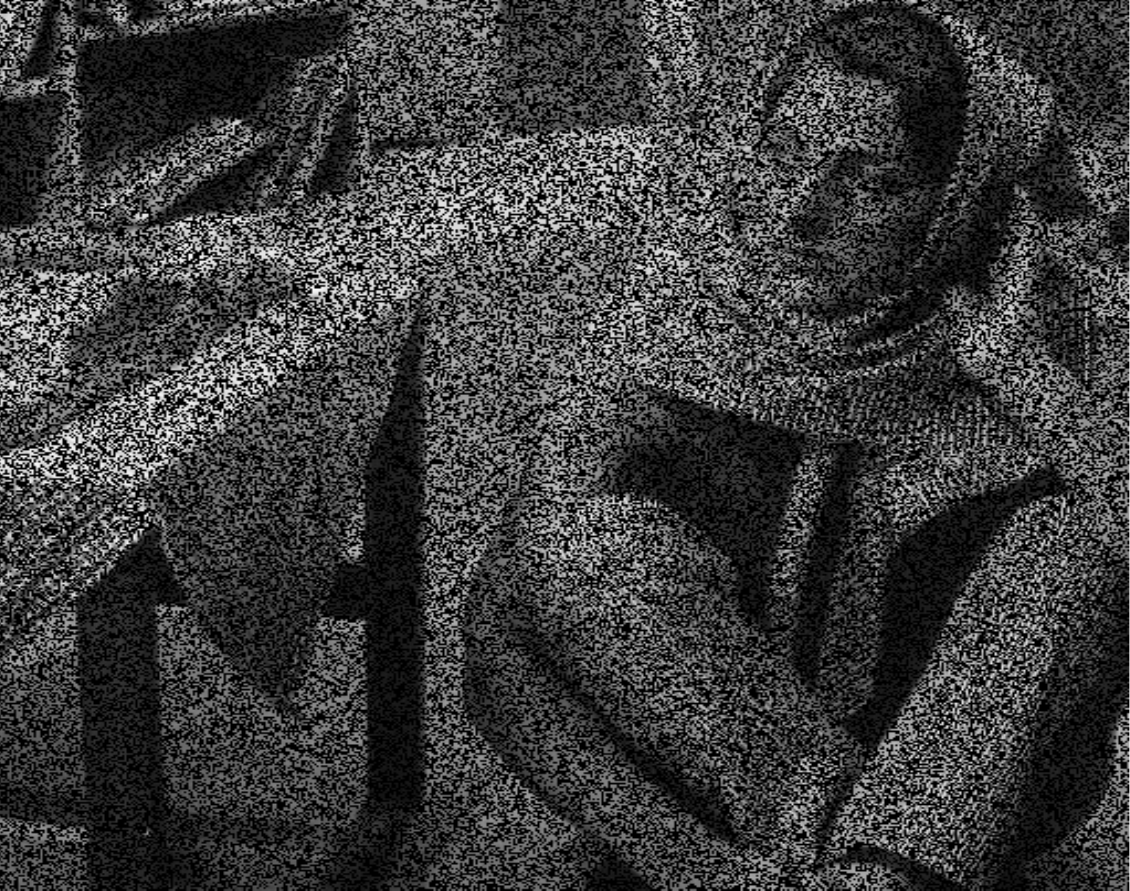} 
    \includegraphics[width=2in]{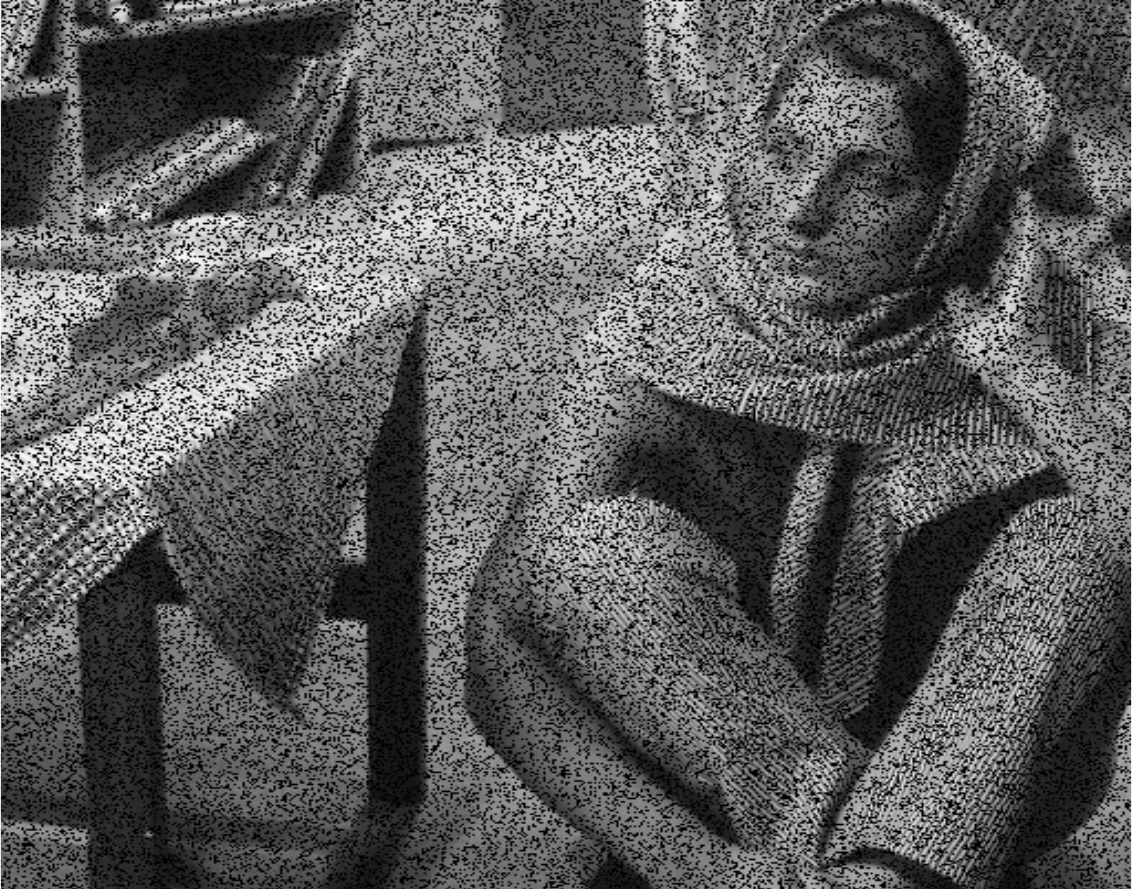} 
    \end{tabular}
    \caption{Top: original $512\times 512$ ``Barbara" image. Bottom, left to right: test versions with 25\%, 50\% and 75\% observed pixel values (the remainder are removed).}
    \label{fig:barb}
\end{figure*}

\begin{figure*}[htb!]
    \centering
    \begin{tabular}{c}
    \includegraphics[width=4in]{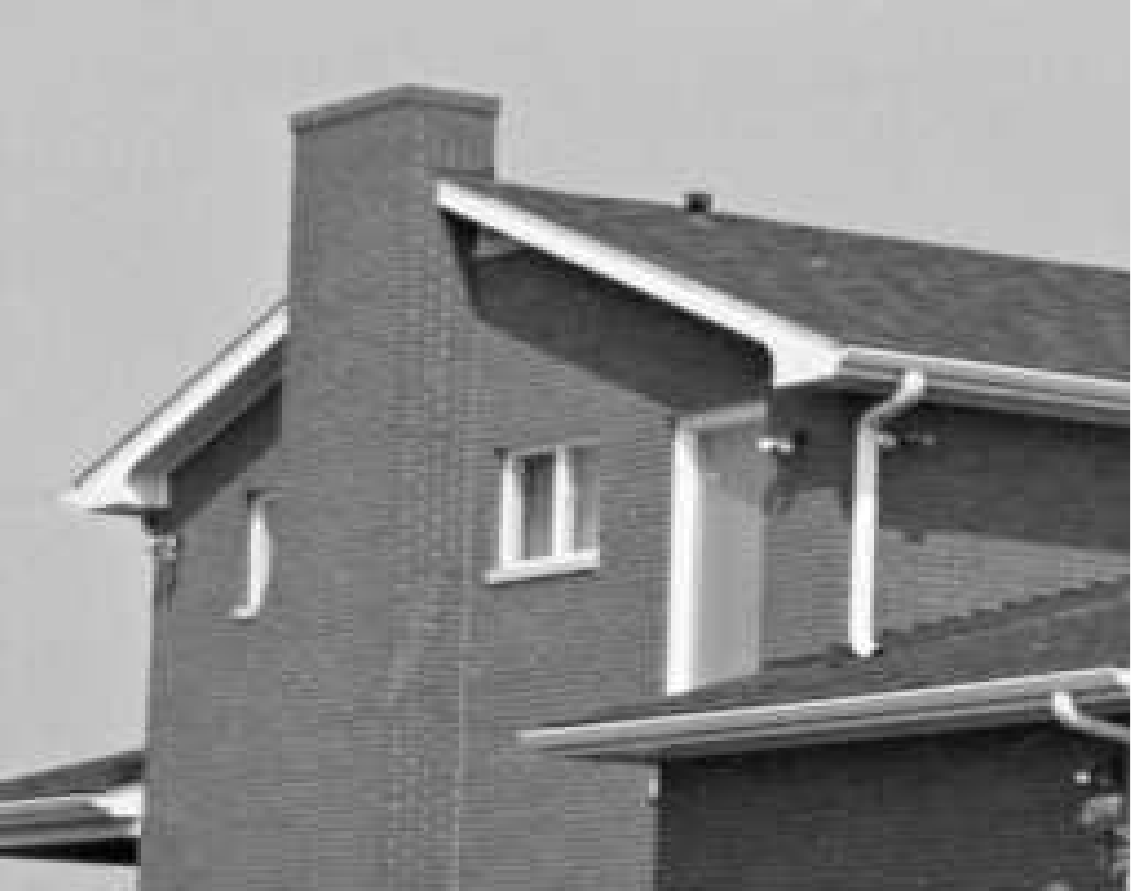} \\
    \includegraphics[width=2in]{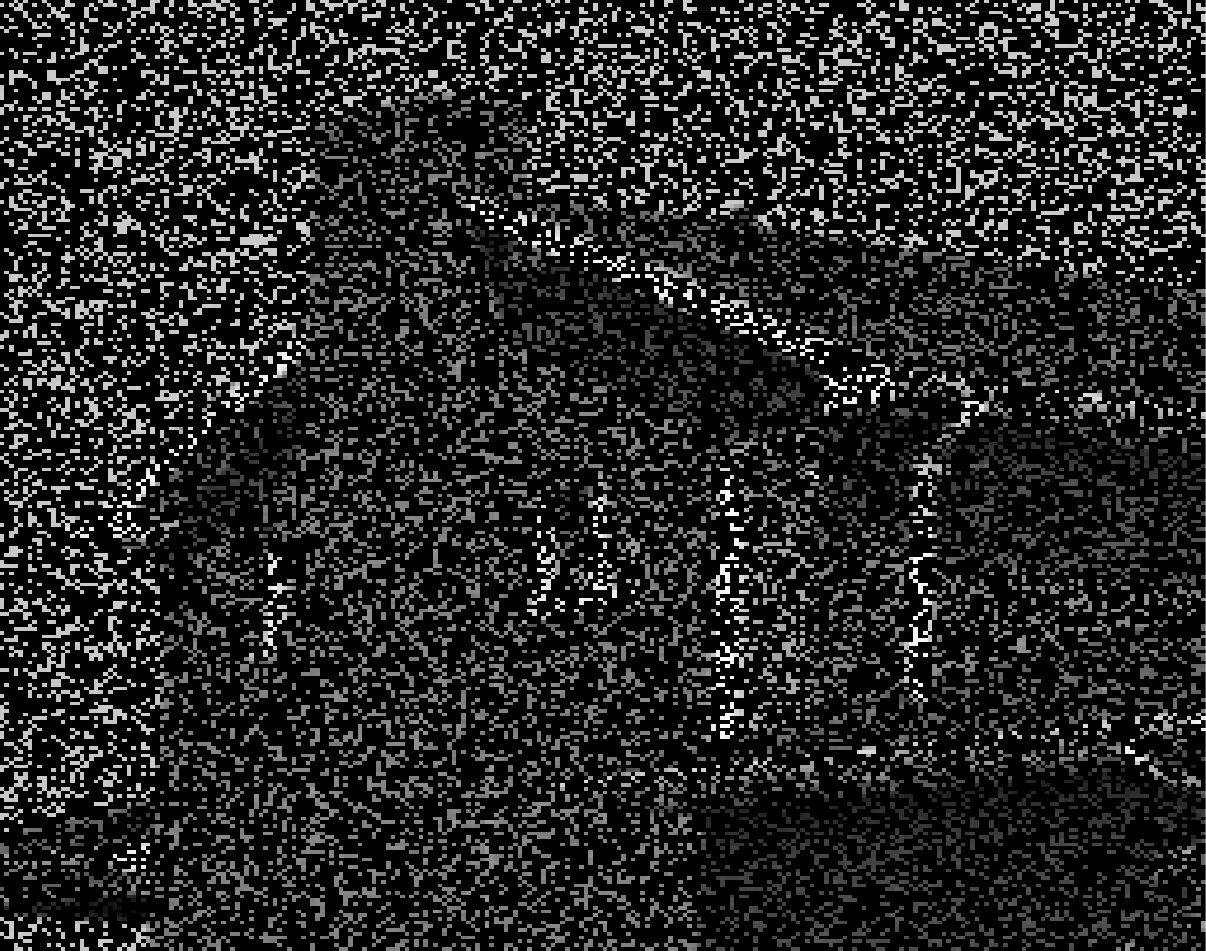} 
    \includegraphics[width=2in]{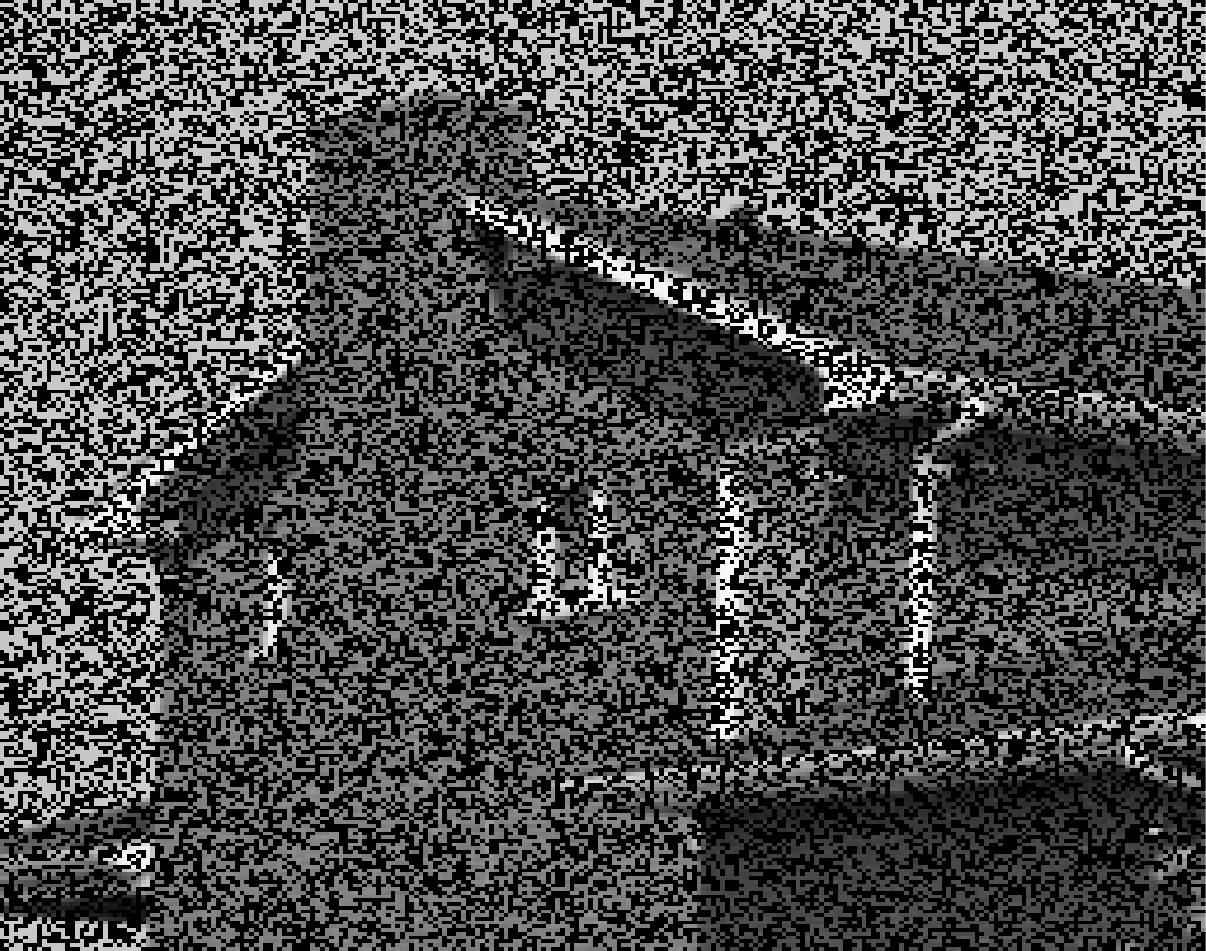} 
    \includegraphics[width=2in]{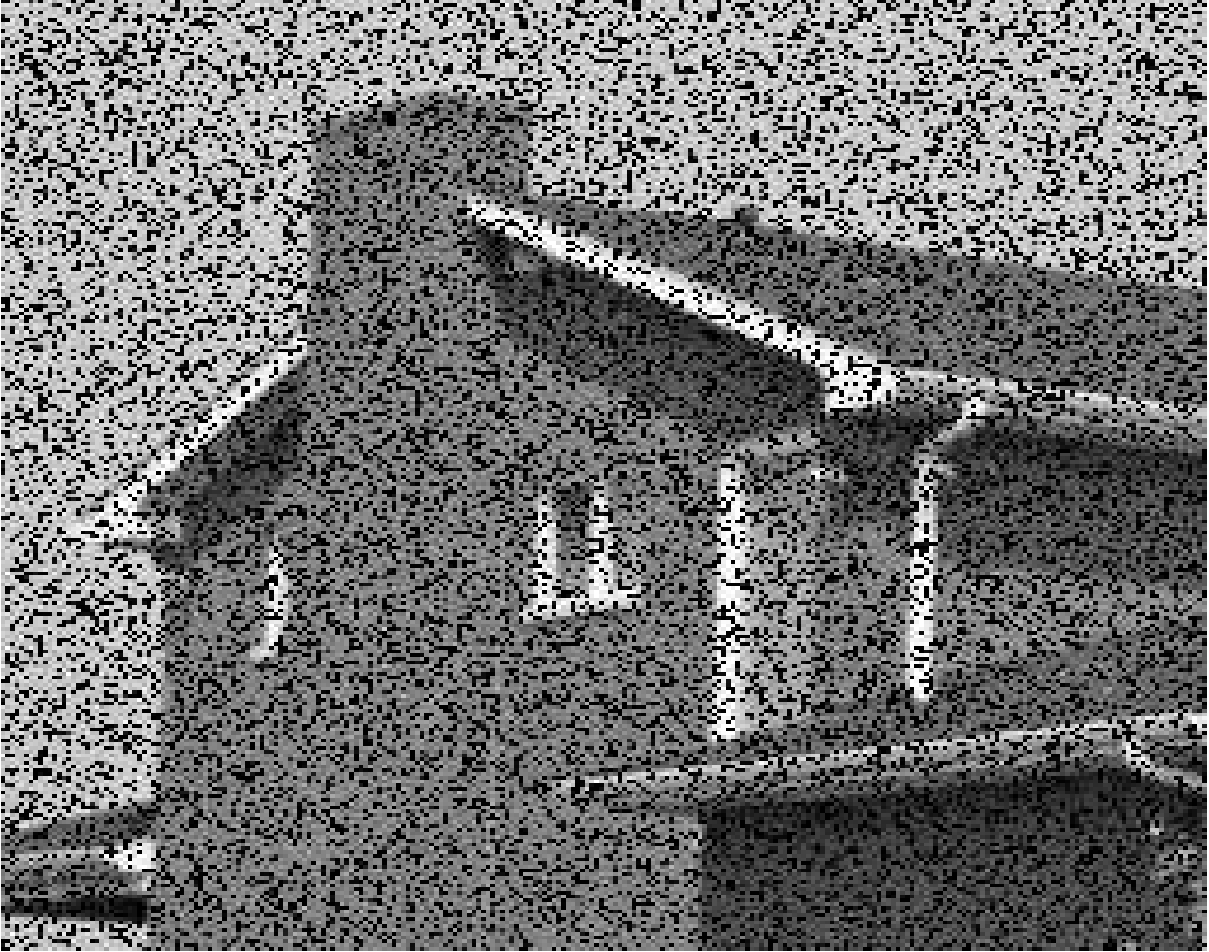} 
    \end{tabular}
    \caption{Top: original $256\times 256$ ``house" image. Bottom, left to right: test versions with 25\%, 50\% and 75\% observed pixel values (the remainder are removed).}
    \label{fig:house}
\end{figure*}

\begin{figure*}[htb!]
    \centering
    \includegraphics[width=4in]{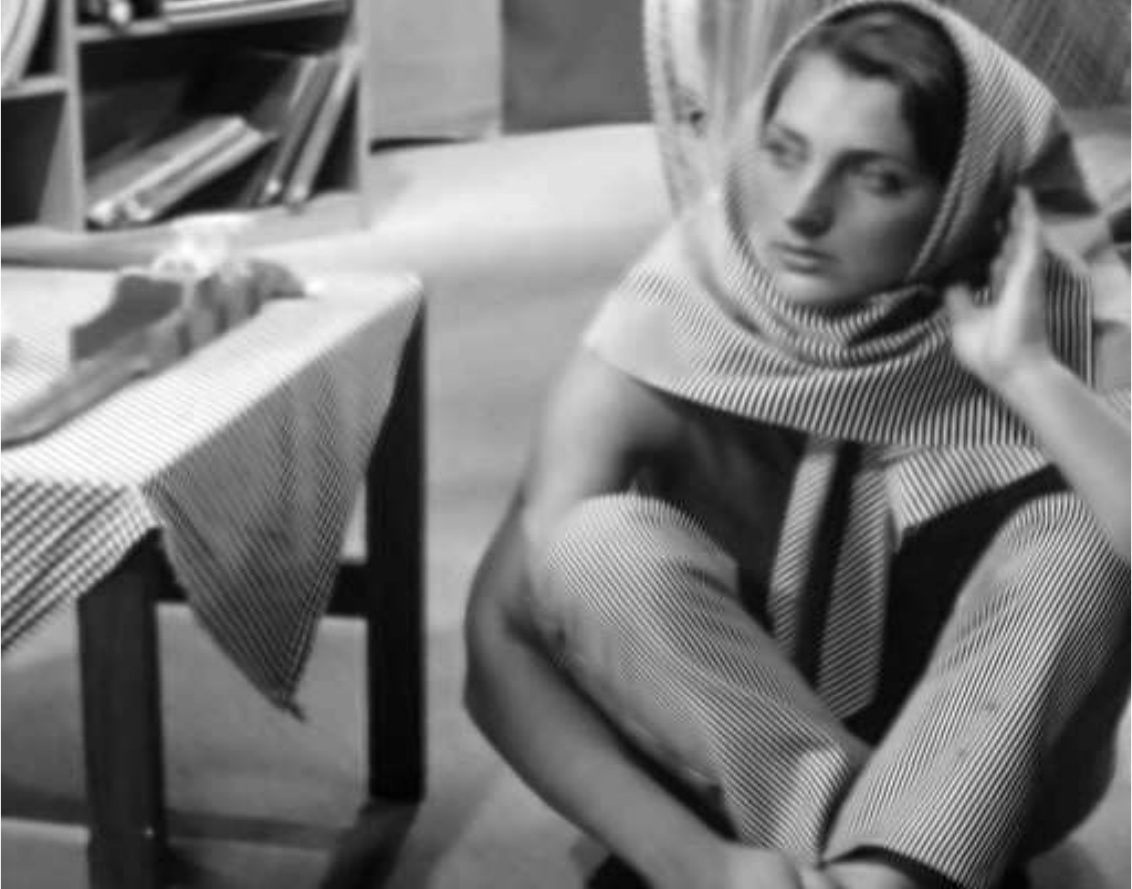}
    \caption{Inpainted $512\times 512$ ``Barbara" image with 50\% observed pixels. The peak signal--to--noise ratio (PSNR) of this estimate is 27.93 dB. Maximum block size ($k_{\max}$) is eight and the total of dictionary elements is set to $r=256$, divided in $L=32$ blocks. Image patches have size $8\times 8$ and are treated as vectors of dimension $n=64$. As we employ overlapping patches, the total number of vectors is $N=255,025$.}
    \label{fig:barbara512_8}
\end{figure*}

\begin{figure*}[htb!]
    \centering
    \includegraphics[width=4in]{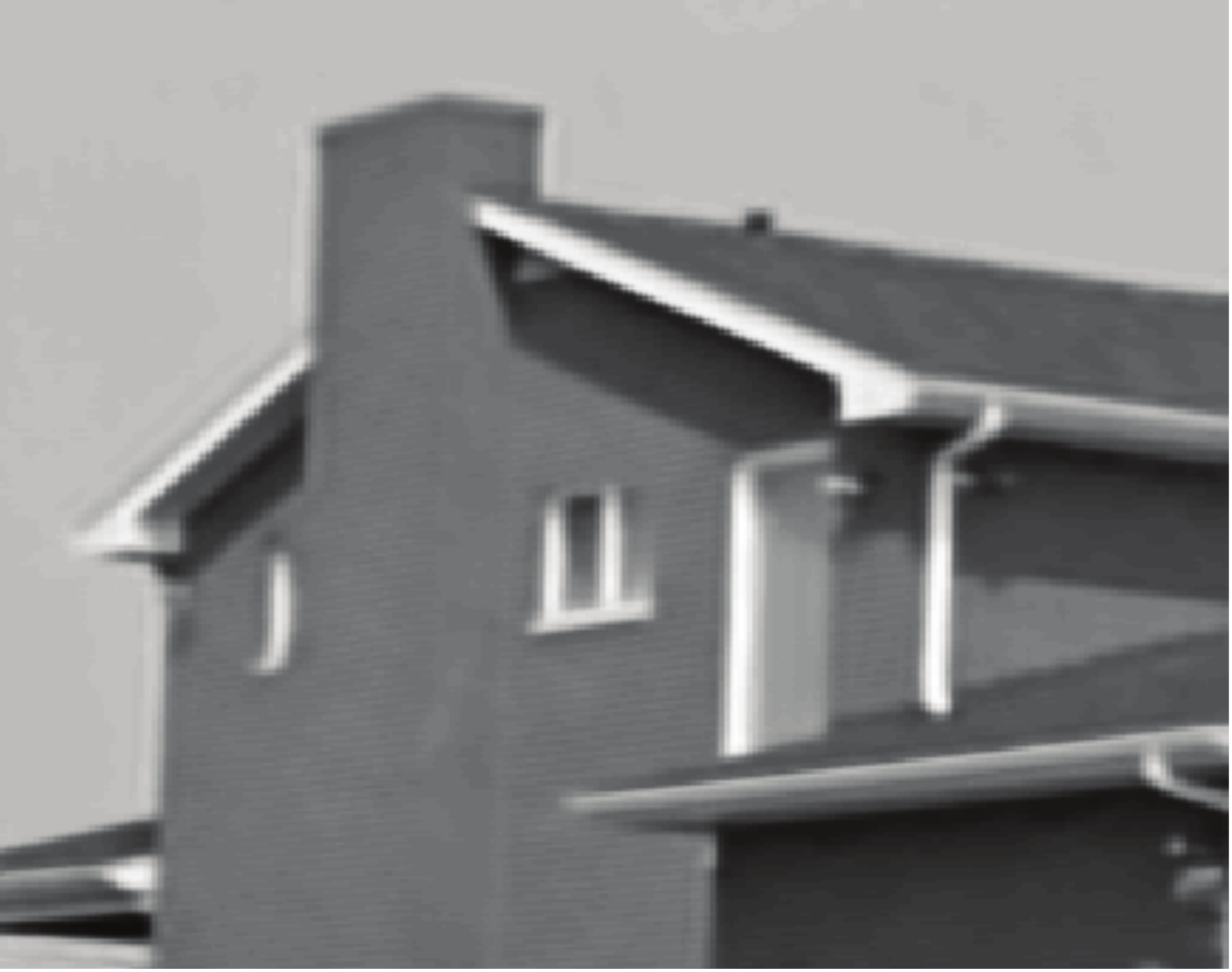}
    \caption{Inpainted $256\times 256$ ``House" image with 50\% observed pixels. The peak signal--to--noise ratio (PSNR) of this estimate is 31.80 dB. Maximum block size ($k_{\max}$) is four and the total of dictionary elements is set to $r=256$, divided in $L=64$ blocks. Image patches have size $8\times 8$ and are treated as vectors of dimension $n=64$. As we employ overlapping patches, the total number of vectors is $N=62,001$.}
    \label{fig:house_4}
\end{figure*}

\begin{figure*}[htb!]
    \centering
    \includegraphics[width=3in]{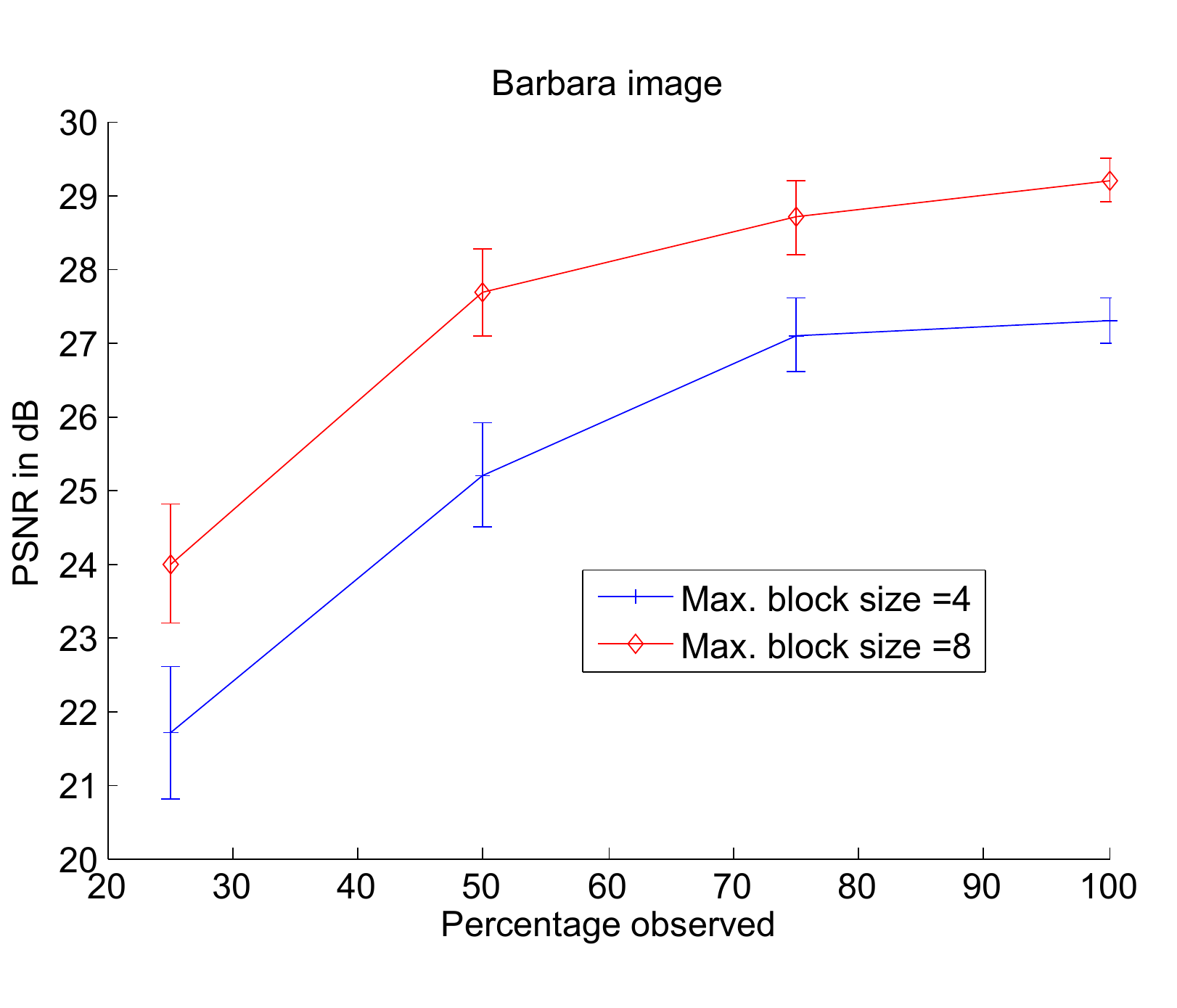} \includegraphics[width=3in]{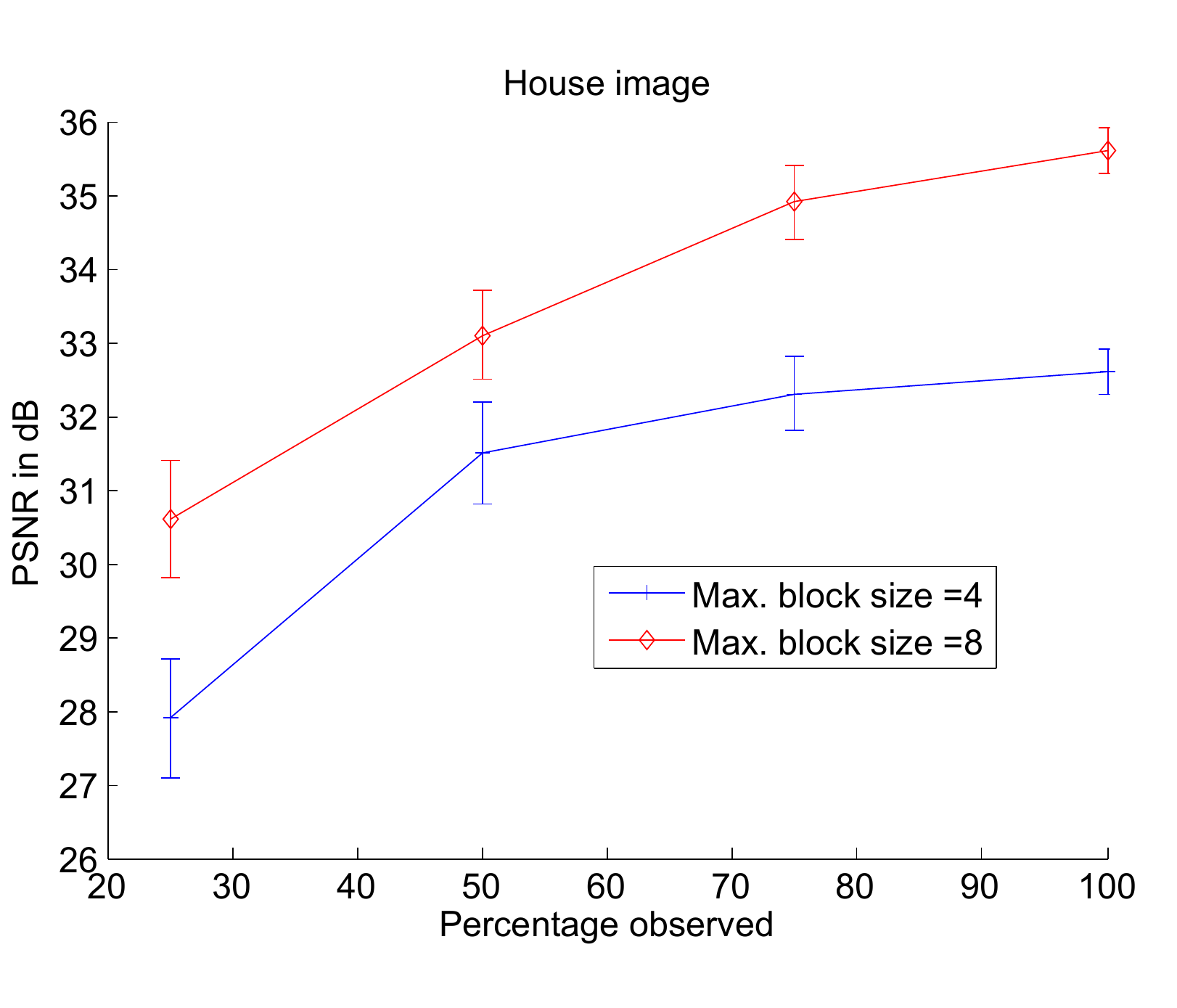}
    \caption{Peak signal--to--noise ratio (PSNR) for inpainting the ``Barbara'' ({\bf left}) and ``house'' ({\bf right}) images with 25\%, 50\%, 75\% and 100\% observed pixel values, using maximum block size ($k_{\max}$) of four and eight. Results are averaged over ten realizations of the randomly missing pixel locations. The error bars depict one standard deviation. Using higher $k_{\max}$ yields better PSNR (but also a less parsimonious model).}
    \label{fig:psnr}
\end{figure*}

\begin{figure}[htb!]
    \centering
    \includegraphics[width=3in]{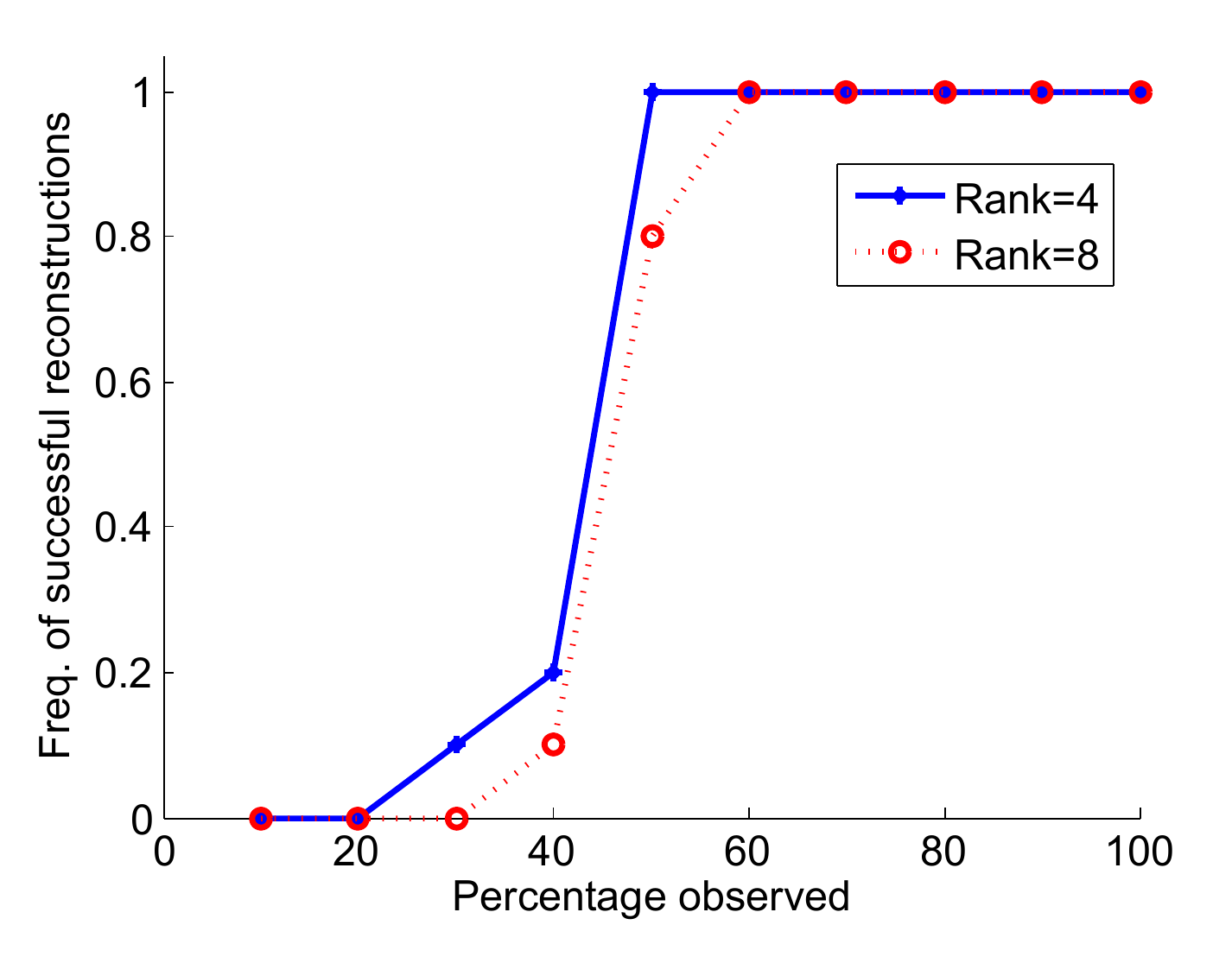} \\
    \caption{Illustration of the phase transition phenomenon, using as input subsets of pixels from an approximation of the ``Barbara'' image that exactly obeys the union--of--subspaces model, with block size (rank) four and eight. The plot shows the normalized frequency of successful reconstructions over ten realizations of missing pixel locations, for varying percentages of observed pixels. We declare a successful reconstruction when the PSNR of the estimate is greater than 40 dB. }
    \label{fig:phase}
\end{figure}

\begin{figure*}[htb!]
    \centering
    \begin{tabular}{c}
    \includegraphics[width=6in]{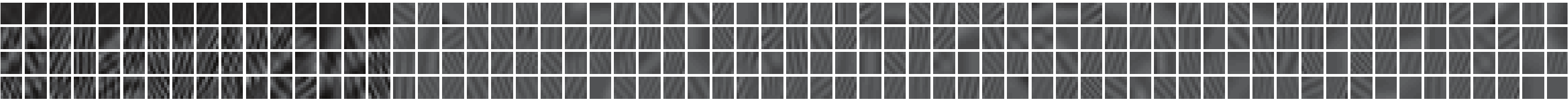} \\
    (a)\\
    \includegraphics[width=3in]{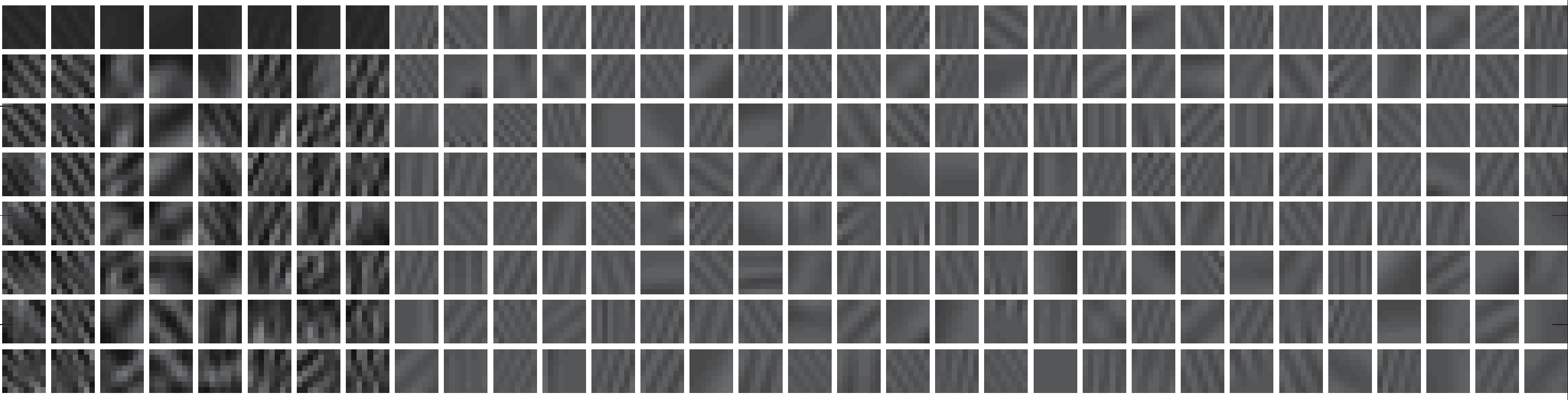}\\
    (b)
    \end{tabular}
    \caption{Learned dictionaries with Algorithm 1 for the ``Barbara'' image using $k_{\max}=4$ (a) and $k_{\max}=8$ (b). Each atom is shown as a $8\times 8$ image. The atoms are arranged in a $k_{\max} \times L$ matrix, with all atoms in each column belonging to the same dictionary block.}
    \label{fig:dict_barbara}
\end{figure*}

\begin{figure*}[htb!]
    \centering
    \begin{tabular}{c}
    \includegraphics[width=6in]{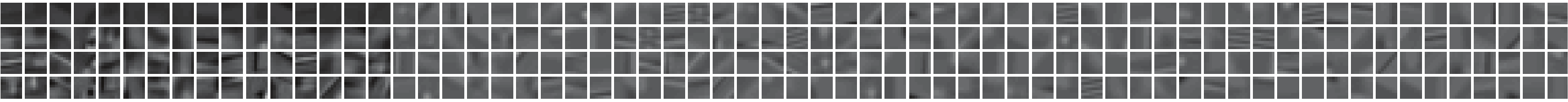} \\
    (a)\\
    \includegraphics[width=3in]{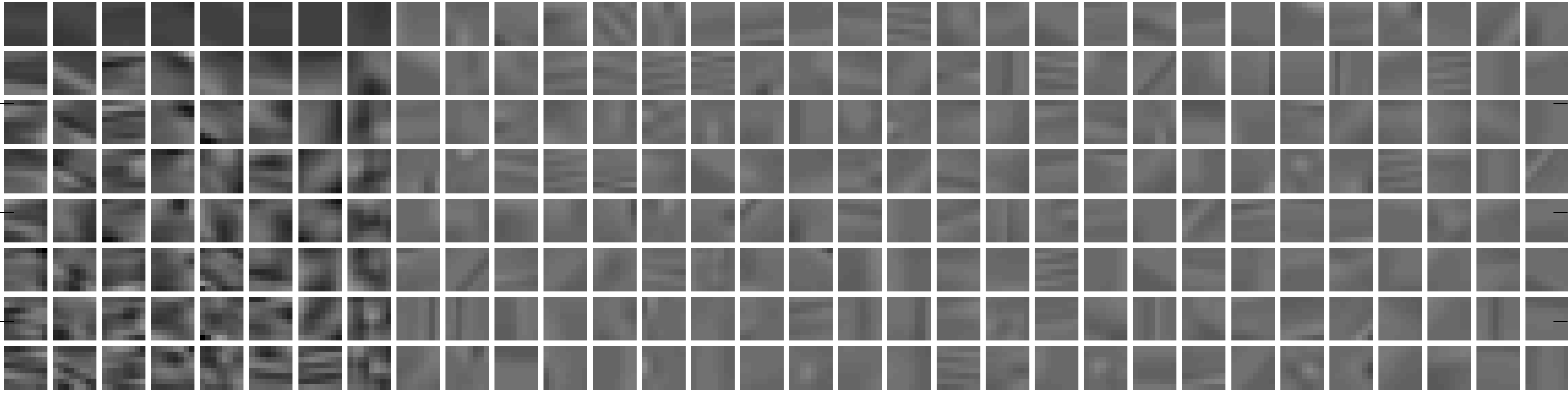}\\
    (b)
    \end{tabular}
    \caption{Learned dictionaries for the ``house'' image using $k_{\max}=4$ (a) and $k_{\max}=8$ (b). Each atom is shown as a $8\times 8$ image. The atoms are arranged in a $k_{\max} \times L$ matrix, with all atoms in each column belonging to the same dictionary block.}
    \label{fig:dict_house}
\end{figure*}

\begin{figure*}[htb!]
    \centering
    \includegraphics[width=4in]{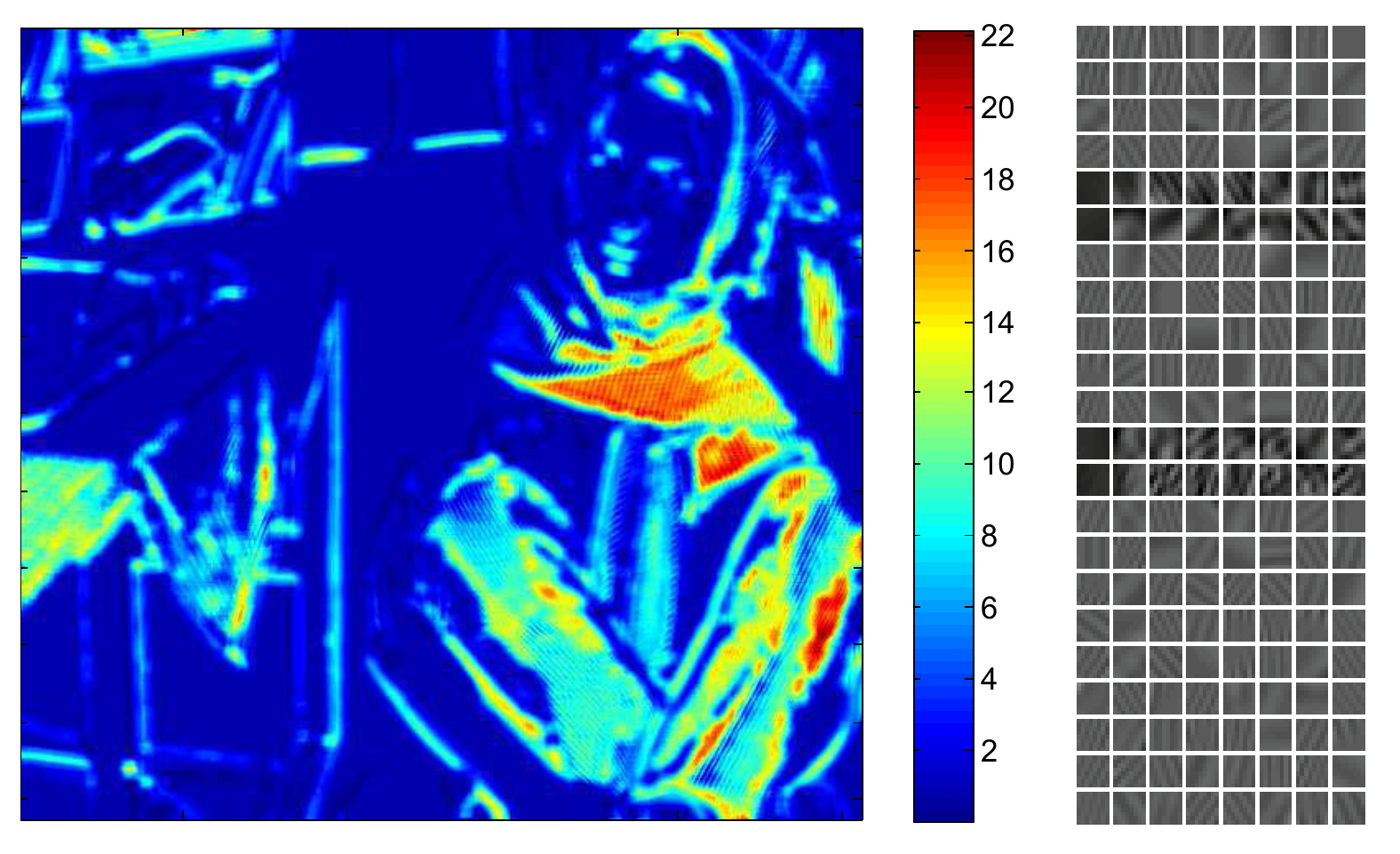} \\
    \caption{Dictionary block usage for the ``Barbara'' image with $k_{\max}=8$. The colors represent the index of the block used by each patch, for a subset containing the 22 most frequently used blocks, out of a total of 32. The subset is shown to the right, aligned with the indices of the colorbar, with each block corresponding of a row of subimages (atoms). It is apparent that patches with similar textures (\emph{e.g.}, the pants, portions of the scarf and table cloth) tend to share blocks (best viewed in color).}
    \label{fig:usage}
\end{figure*}

\section{Conclusion}
\label{sec:conclusion}
We have proposed a framework for simultaneous estimation of one--block--sparse signals and the corresponding dictionary, based on compressive measurements. Multiple sensing matrices are employed, which allows use of low--rank matrix completion results to guarantee unique recovery (up to a specified equivalence class) for a sufficiently large number of signals and measurements per signal; bounds are derived for the number of measurements. The assumption of one--block sparsity is related to the union--of--subspaces signal model and to the well--known mixture of factor analyzers (MFA) statistical framework. Existing results for the related problems of dictionary learning (DL) and block--sparse dictionary optimization (BSDO) are extended, due to the analysis of compressive measurements, and blind compressed sensing (blind CS) is also extended by considering a broader class of dictionaries. 

Additionally, a practical algorithm based on alternating least--squares minimization has been proposed. The algorithm is related to BSDO, and has been shown to converge to a local optimum under mild conditions. We observe encouraging performance on image inpainting tasks without the need for parameter tuning or careful initialization. An avenue for further research is the treatment of observation noise, which constitutes a natural extension of our work, as well as extension of the theory beyone one--block sparsity.
%, and the enforcement of the block--orthonormality constraint at all steps of the algorithm, through formulation as an optimization problem on a Stiefel manifold.

\bibliography{bsdd}
\bibliographystyle{plain}
\end{document}